\documentclass[12pt]{article}
\usepackage[margin=0.9in]{geometry}
\usepackage{amsthm,amsmath,amssymb,amscd,enumerate, graphicx,caption,subcaption,bm,booktabs,longtable,colortbl,array,makecell}
\usepackage{setspace,array}
\usepackage{thmtools}
\usepackage{thm-restate}
\usepackage{color}
\usepackage{xcolor}
\usepackage{natbib}
\usepackage[mathscr]{euscript}
\usepackage[shortlabels]{enumitem}
\usepackage[colorlinks=true, urlcolor=blue, linkcolor=blue, citecolor=blue]{hyperref}
\usepackage{xr-hyper}
\usepackage{xr}
\usepackage{algpseudocode}
\usepackage{algorithm}
\algrenewcommand\algorithmicindent{0.4em}
\usepackage{natbib}
\usepackage[T1]{fontenc}
\usepackage[utf8]{inputenc}
\usepackage{authblk}
\usepackage{threeparttable}
\usepackage{multirow}

\DeclareMathOperator{\expit}{expit}
\DeclareMathOperator*{\argmin}{arg\,min}

\newcommand{\indep}{\mathrel{\perp\!\!\!\perp}}

\newtheorem{lemma}{Lemma}
\newtheorem{theorem}{Theorem}

\title{Marginal generalized raking with parametric working models}
\author[1,2,3,*]{Brian~D. Williamson}
\author[2]{Runjia Zou}
\author[4]{Thomas Lumley}
\author[5]{Bryan~E. Shepherd}
\author[1,2]{Pamela~A. Shaw}
\affil[1]{Biostatistics Division, Kaiser Permanente Washington Health Research Institute}
\affil[2]{Department of Biostatistics, University of Washington}
\affil[3]{Vaccine and Infectious Disease Division, Fred Hutchinson Cancer Center}
\affil[4]{Department of Statistics, University of Auckland}
\affil[5]{Department of Biostatistics, Vanderbilt University}
\affil[*]{Corresponding Author: Brian D. Williamson. Email: brian.d.williamson@kp.org}

\makeatletter
\newcommand*{\addFileDependency}[1]{% argument=file name and extension
  \typeout{(#1)}% latexmk will find this if $recorder=0 (however, in that case, it will ignore #1 if it is a .aux or .pdf file etc and it exists! if it doesn't exist, it will appear in the list of dependents regardless)
  \@addtofilelist{#1}% if you want it to appear in \listfiles, not really necessary and latexmk doesn't use this
  \IfFileExists{#1}{}{\typeout{No file #1.}}% latexmk will find this message if #1 doesn't exist (yet)
}
\makeatother

\newcommand*{\myexternaldocument}[1]{%
    \externaldocument{#1}%
    \addFileDependency{#1.tex}%
    \addFileDependency{#1.aux}%
}
\myexternaldocument{manuscript/supp}

\newif\ifarXiv
\arXivtrue 
\begin{document}

\maketitle

\begin{abstract}
    Generalized raking (GR) was originally developed in the survey statistics literature to incorporate auxiliary information in estimation. Recently, it has been used in the biostatistical and epidemiological literature to estimate regression coefficients in parametric models in cases with missing data, including missing data by design (e.g., two-phase studies). In the regression parameter context, the optimal GR estimator has been shown to be equivalent to the optimal augmented inverse probability weighted estimator. In this paper, we generalize the influence function-based theory for GR to marginal estimands; we call our approach \textit{marginal generalized raking}. We contrast our approach to a procedure that marginalizes a conditional GR estimator of regression parameters; while the former approach is doubly robust, i.e. is consistent if either the missing-data model or the optimal raking variable is estimated consistently, our proposed approach is multiply robust, allowing either the outcome regression or treatment propensity score to also be misspecified. We compare the approaches in both fully-synthetic simulations and in an application using data from an observational cohort of persons living with HIV. 
\end{abstract}

\doublespacing

\section{Introduction}\label{sec:intro}

In the biomedical literature, there is increasing interest in estimating target parameters that are marginal quantities (e.g., average treatment effects) rather than conditional quantities (e.g., regression coefficients) \citep{scharfstein1999adjusting,bang2005doubly,kang2007demystifying,vanderlaan2006targeted,cao2009improving,vanderlaan2011targeted,shook2025double,williamson2026assessing}. Marginal parameters are natural targets in causal inference studies. Under certain assumptions, the estimable quantity is a causal effect \citep{vanderlaan2006targeted}; the marginal parameter can still be a valid statistical target regardless of whether the causal assumptions hold. 

Estimation is complicated in settings with missing or mismeasured data. We will refer to such settings as coarsened data settings \citep{vandervaart2000asymptotic} and propose to use methods from the two-phase sampling literature to handle the missing data, namely generalized raking \citep{deville1993generalized,breslow2009improved,lumley2011connections}. In the measurement error setting, the missing data are gold standard (error-free) data that are only measured on a subsample. Classical approaches for estimating marginal or conditional estimands in the coarsened data setting are multiple imputation \citep[MI;][]{rubin2018multiple} and inverse probability of coarsening weighting \citep[IPCW;][]{horvitz1952generalization}, but both approaches have drawbacks. MI relies heavily on the imputation model, which must be specified carefully to avoid bias \citep{williamson2026assessing} and to avoid invalid variance estimates \citep{robins2000inference}. IPCW relies heavily on the missing-data model and can be less efficient than MI \citep{little2024comparison}. Augmented IPCW \citep[AIPCW;][]{robins1994estimation,robins1995semiparametric} was designed to reduce the reliance of IPCW on correctly specifying the missing-data model: instead, an AIPCW estimator is consistent if either the missing-data model or a projection of the efficient influence function \citep[EIF;][]{bickel1993efficient} onto the fully observed data is estimated consistently. In other words, it is doubly robust \citep{seaman2018introduction}. However, the AIPCW estimator can lie outside the bounds of the parameter space \citep{diaz2020machine,ellul2024causal}. Generalized raking (GR), where the IPCW weights are calibrated using auxiliary variables, was originally developed in the survey literature to estimate population totals \citep{deville1992calibration,deville1993generalized}; an influence function-based GR procedure for regression parameters has been shown to be asymptotically equivalent to the optimal AIPCW estimator for regression coefficients if the proper calibration variables are chosen \citep{breslow2009improved,breslow2009using,lumley2011connections} and has the advantage of remaining in the target parameter space. Finally, targeted maximum likelihood estimation \citep[TMLE;][]{vanderlaan2006targeted,vanderlaan2011targeted} and its missing-data generalization IPCW-TMLE \citep{rose2011targeted} can also be used to estimate both marginal and conditional estimands and make less restrictive assumptions on the data-generating process than parametric approaches. Both GR and IPCW-TMLE respect the bounds of the parameter space. 

\citet{williamson2026assessing} recently compared GR to IPCW and IPCW-TMLE in a wide variety of synthetic data and plasmode simulations. They found that GR was often more efficient than IPCW-TMLE, due to both making parametric modeling assumptions and calibration of the weights. Because the natural target of current GR estimators in the biomedical literature is a regression coefficient, \citet{williamson2026assessing} first estimated regression coefficients using GR and then marginalized (i.e., integrated over the covariate distribution) to estimate marginal quantities, relying on the delta method for inference. Importantly, this procedure required the outcome model to be correctly specified for consistent estimation of marginal parameters \citep[][Table 1]{williamson2026assessing} and the variance estimator was used without proving that it was valid. We hypothesize that by using the EIF for the marginal estimand of interest directly, we can remove the need to estimate regression parameters first and can increase the robustness of raking estimators. In this article, we show that our proposed estimator, which we call \textit{marginal} generalized raking, is asymptotically equivalent to the AIPCW estimator of a marginal estimand and is efficient when using parametric working models. This is appealing because there are straightforward implementations of GR in standard software \citep{surveypkg}. Further, we show that the marginal GR estimator is multiply robust \citep{tchetgen2012semiparametric}, i.e., consistent under weaker conditions than the original GR estimator of a marginal estimand. 

The remainder of this manuscript is organized as follows. In Section~\ref{sec:raking_marginal}, we provide an overview of generalized raking for regression coefficients and introduce two influence function-based generalized raking methods for estimating marginal targets and describe their asymptotic properties. In Section~\ref{sec:sims}, we evaluate the performance of these estimators in numerical experiments. We analyze data from an observational cohort of persons living with HIV in Section~\ref{sec:data_analysis}  and provide concluding remarks in Section~\ref{sec:conclusions}. Technical details, including proofs of theorems, and additional data analysis results are provided in the Supplementary Material. Numerical experiments and data analyses can be reproduced using code available online at \url{https://github.com/bdwilliamson/mgr_parametric_supplementary}.

\section{Influence function-based generalized raking for marginal estimands}\label{sec:raking_marginal}

\subsection{Data structure and notation}

Throughout, we will refer to the target parameter as an estimand. Suppose that the ideal data are $D := (L, A, X, Y)$, where $L \in \mathbb{R}^p$ are possible confounding variables, $A \in \mathbb{R}^d$ are possible auxiliary variables, $X \in \{0, 1\}$ is a binary exposure, and $Y \in \mathbb{R}$ is the outcome of interest. Rather than observing $D$ on each of $n$ individuals, we instead observe a coarsened set of data $O$, with indicator $R \in \{0,1\}$ of having $D$ observed. We consider two examples to make this discussion more concrete. In a setting with missing confounders, which we will illustrate in our simulations below, we might observe $O = (Z, R, R W, X, Y)$, where $L = (Z,W)$, $Z$ are always observed confounders, $W$ are sometimes-observed confounders, $R$ is an indicator of whether $W$ is observed. In our observational cohort study that we analyze in Section~\ref{sec:data_analysis}, there are several variables that are measured without error ($Z$) and others, including the outcome and exposure, that are measured with error. In this scenario, we could consider auxiliary variables $A = (W^*,X^*,Y^*)$ consisting of the error-prone versions of the variables. Then we observe $O = (Z,A,R, RW, RX, RY)$. Throughout, we will use $V$ to refer to the collection of variables that are always observed (excluding $R$). We will assume that the data are missing at random, and that $V$ encompasses the variables that are necessary for estimating the probability of having complete data. We will use the potential (or counterfactual) outcomes notation $Y(x)$ to denote the outcome under treatment $x$. We define the outcome regression function $Q(x,l) := E(Y \mid X = x, L = l)$, which in this paper we assume to follows a parametric generalized linear model, i.e., that $Q(x,l) = h^{-1}\{ (1,x,l)\beta\}$ for link function $h$ and regression parameter vector $\beta \in \mathbb{R}^{p+1}$. We also define the probability of having complete data $\pi(v) := P(R = 1 \mid V = v)$ and the treatment-assignment probability, or propensity score, $g(l) = P(X = 1 \mid L = l)$. Finally, we assume that the data are missing at random, i.e., $R \indep D \mid V$.  

\subsection{Target estimands}\label{sec:targets}

Suppose that we have a sample $O_1,\ldots, O_n$ arising from a standard two-phase setting where some variables are available on a phase 1 cohort of size $n$ and additional variables are available on a phase 2 subsample of size $n_2 < n$ (i.e., those with $R = 1$). Our goal is to use $O_1,\ldots, O_n$ to make inference on several target estimands. In settings with a binary exposure variable, we can define the treatment-specific mean outcome values $\mu_x := E\{Y(x)\}$ and the average treatment effect $ATE:= \mu_1 - \mu_0$. In binary outcome settings, the ATE corresponds to the risk difference (RD) and we can further define the relative risk (RR) as a function of $\mu_1$ and $\mu_0$: $RR:= \mu_1 / \mu_0$.

Under the common causal inference assumptions of consistency (i.e., the observed outcome for observation $i$ is the potential outcome under the observed treatment/exposure), no unmeasured confounding, and positivity (i.e., $\pi(v) > 0$ for all $v$) \citep[e.g.,][]{cole2009consistency,vanderlaan2011targeted}, we can identify these estimands with the ideal data; writing
\begin{align*}
    \mu_1 =& \ E\{E(Y \mid X = 1, L)\} \text{ and } \\
    \mu_0 =& \ E\{E(Y \mid X = 0, L)\}
\end{align*}
identifies the ATE and RR. If the causal inference assumptions do not hold, these remain valid statistical estimands but cannot be causally interpreted. Using the observed data, these can be identified by 
\begin{align*}
    \mu_1 =& \ E\left\{\frac{R}{\pi(V)}E(Y \mid X = 1, L)\right\} \text{ and } \\
    \mu_0 =& \ E\left\{\frac{R}{\pi(V)}E(Y \mid X = 0, L)\right\}.
\end{align*}
Other identification formulas are possible: for example, $\mu_1 = \ E\left\{\frac{R}{\pi(V)}\frac{I(X = 1)}{g(L)}Y\right\}$ uses inverse probability of treatment weighting, where $g$ is the true treatment propensity score.

\subsection{Generalized raking for regression parameters}\label{sec:gr_regression}
In biostatistics, following the landmark work of \citet{breslow2009improved,breslow2009using}, generalized raking has been traditionally used to estimate coefficients in regression models. We refer to these as \textit{conditional} parameters, because their interpretation is conditional on the values of covariates in the model. In this setting, the optimal generalized raking estimator has been shown to be equivalent to the optimal augmented inverse probability-weighted estimator \citep{lumley2011connections}. Our first approach to estimating a marginal target estimand (Section~\ref{sec:cgr}) will use generalized raking for regression parameters as a preliminary step, so we briefly describe the traditional procedure here.

We will use generalized raking to estimate the regression parameters $\beta$ using the observed data. First, we obtain an estimator $\pi_{n,IPCW}$ of $\pi$, the probability of having complete data; often, logistic regression is used. Throughout, we use the subscript $n$ to denote a quantity that is estimated from the data. We use the subscript $IPCW$ to emphasize that these are inverse probability of coarsening weights. Next, we calibrate the weights. The ideal calibration uses the optimal raking variable $\eta_\beta(v) = E\{\phi_\beta^F(D) \mid V = v\}$ \citep{lumley2011connections,breslow2009improved}, where $\phi_\beta^F$ is the vector of EIFs for the regression parameters in the ideal-data setting (no missing data). This is the same as the optimal choice of augmentation term in augmented IPCW \citep[AIPCW;][]{robins1994estimation,robins1995semiparametric}. The optimal raking variable can be estimated in several ways, including regression \citep{breslow2009improved} and single imputation \citep{kulich2004improving,breslow2009using}; in the measurement error context, one can use the error-prone variables to estimate this variable \citep{kulich2004improving, breslow2009using, shepherd2023multiwave}. We can also estimate $\eta_\beta$ using multiple imputation with $M$ imputations \citep{han2016combining,oh2021improved,han2021combining}. Specifically, after specifying an imputation model, for each of the $M$ imputations we: (i) impute the phase 2 variables for all observations (including those with complete data); (ii) fit the specified regression model to the completed dataset with only imputed phase 2 variables; and (iii) extract the influence function values, which can be achieved using the outputs of standard regression software. The final multiply-imputed estimated efficient influence function values $\eta_{\beta,n}(v_i)$ for each observation are then the average across the imputations. Regardless of how the optimal raking variable is estimated, we use it to obtain the calibrated  weights $\pi_{n,GR,i}^{-1}$ as follows. For a distance measure $d$, 
\begin{align}
    \pi_{n,GR,i}^{-1} =& \ b_i^*\pi_{n,IPCW,i}^{-1}, \text{ where } \label{eq:gr}\\
    b^* =& \ \argmin_{b \in \mathbb{R}^n}\sum_{i=1}^n R_i d(b_i\pi_{n,IPCW,i}^{-1},\pi_{n,IPCW,i}^{-1}) \notag \\
    & \text{ subject to }  \sum_{i=1}^n \eta_{\beta,n}(V_i) = \sum_{i=1}^n R_i b_i \pi_{n,IPCW,i}^{-1} \eta_{\beta,n}(V_i). \notag
\end{align}
A common choice of distance measure for two scalars $a$ and $c$ is the Poisson deviance $d(a,c) = a \log (a/c) - a + c$, which results in non-negative final weights \citep{lumley2011connections}. The final estimator $\beta_{n,GR}$ of $\beta$ is the solution to the estimating equation
\begin{align*}
    \frac{1}{n}\sum_{i=1}^n\frac{R_i}{\pi_{n,GR}(V_i)}S(Y_i,X_i,L_i;\beta) = 0,
\end{align*}
where $S(\cdot;\beta)$ is the score function defined by the target regression model. Importantly, $\sqrt{n}(\beta_{n,GR} - \beta) \to_d N(0, \Sigma_\beta)$, where $\Sigma_\beta = E\{\phi_\beta(Y,X,L) \phi_\beta(Y,X,L)^\top\}$ \citep{lumley2011connections}, under the positivity assumption on $\pi$, $n_1/n \to C \in (0,\infty]$,  and the assumption that $\pi_n \to_P \pi$ or $\eta_{\beta,n} \to_P \eta_\beta$; if both estimators $\pi_n$ and $\eta_{\beta,n}$ are consistent, then $\beta_{n,GR}$ is efficient. In other words, $\beta_{n,GR}$ has the smallest asymptotic variance among the class of design-based estimators \citep{lumley2011connections}, which includes the optimal AIPCW estimator \citep{lumley2011connections, robins1995semiparametric}. The efficient influence function vector $\phi_\beta$ is a function of the ideal-data EIF \citep{vandervaart2000asymptotic}, specifically,
\begin{align*}
    \phi_\beta(y,x,l,r,v) := \frac{r}{\pi(v)}\phi_\beta^F(y,x,l) - \frac{\{r - \pi(v)\}}{\pi(v)}\eta_\beta(v).
\end{align*}

\subsection{Conditional generalized raking}\label{sec:cgr}

Our first approach to estimating a marginal target estimand using generalized raking, used in \citet{williamson2026assessing} and which we call \textit{conditional generalized raking} (CGR), consists of two steps. The first step is to use generalized raking to estimate regression parameters for a generalized linear model. In the second step, we marginalize over the distribution of the covariates. 

There are several approaches that can be taken to account for possible confounding of the association between $X$ and $Y$. We will consider two here: regression adjustment (RA) and inverse probability of treatment weighting (IPTW). The RA approach includes covariates $L$ in the outcome regression model used in generalized raking. This is equivalent to the procedure described in Section~\ref{sec:gr_regression} and results in estimator $\beta_{n,GR}$ and calibrated weights $\pi_{n,CGR}$ using Equation~\eqref{eq:gr}. Then we can define outcome regression estimators $Q_{n,1,CGR,RA}(x,l) = Q_{n,0,CGR,RA}(x,l) = h^{-1}\{(1,x,l)\beta_{n,GR}\}$; in other words, the outcome regression model has the same form for $X = 0$ and $X = 1$. The IPTW approach uses weights to account for the probability of receiving treatment, and the final outcome regression model only contains $X$ and an intercept. Because this approach is less common in the parametric regression modeling literature, we describe the IPTW approach more fully in Section~\ref{sec:cgr_iptw}. However, we will compare this estimator to our proposed estimator in the simulations presented below, because both approaches use inverse probability of treatment weights. Regardless of the method that we used to obtain $Q_{n,1,CGR}$ and $Q_{n,0,CGR}$, we can estimate $\mu_1$ and $\mu_0$ using  
\begin{align*}
    \mu_{1,n,CGR} :=& \ \frac{1}{n}\sum_{i=1}^n \frac{R_i}{\pi_{n,CGR}(V_i)} Q_{n,1,CGR}(1,L_i) \text{ and }\\
    \mu_{0,n,CGR} :=& \ \frac{1}{n}\sum_{i=1}^n \frac{R_i}{\pi_{n,CGR}(V_i)}Q_{n,0,CGR}(0,L_i),
\end{align*}
and we can estimate the ATE and RR by plugging these estimators into the formulas defined above. Inference can be carried out using the delta method. 

Before stating our first result, we define several assumptions for a regression adjustment estimator (equivalent assumptions for the IPTW estimator are in Section~\ref{sec:cgr_iptw}):
\begin{itemize}
    \item[(A1)] \textit{(positivity of the missing-data model)} $\pi(v) > \epsilon > 0$ for all $v$;
    \item[(A2)] \textit{(consistency of the regression model)} $Q_{n,CGR} \to_P Q$;
    \item[(A3)] \textit{(consistency of the missing-data model)} $\pi_{n,CGR} \to_P \pi$;
    \item[(A4)] \textit{(consistency of the EIF projection)}  $\eta_{\beta,n} \to_P \eta_\beta$.
\end{itemize}
Assumption (A2) is equivalent to stating that the estimator of the full $(p+1)$-parameter outcome regression model that accounts for confounders $L$ is consistent. Assumption (A4) requires that the estimated full $(p+1)$-parameter optimal raking variable $\eta_{\beta,n}$ is consistent. Assumptions (A2) and (A4) are related: because $\eta_\beta$ is defined with respect to a parametric model, if Assumption (A2) is not satisfied then it is more difficult to satisfy Assumption (A4). IPTW requires a modification of (A2) and (A4), which we refer to as (A2') and (A4'); we discuss these conditions further in Section~\ref{sec:cgr_iptw}. 

With these assumptions in place, we can describe the asymptotic behavior of $\mu_{1,n,CGR}$ using regression adjustment. As mentioned above, this was stated without proof in \citet{williamson2026assessing}.
\begin{lemma}\label{lemma:cgr_variance}
    If (A1), (A2), and either (A3) or (A4) hold, then 
    \begin{align*}
        \sqrt{n}(\mu_{1,n,CGR} - \mu_1) \to_d& \ Z \sim N(0, \sigma^2_{1,CGR}), \text{ where } \\
        \sigma^2_{1,CGR} =& \ \Omega\Sigma_\beta \Omega^\top, \\
        \Omega =& \ E\left\{\frac{R}{\pi(V)}\nabla h^{-1}(1,1,L;\beta)\right\}, \text{ and } \\
        \Sigma_\beta =& \ E\{\phi_\beta(Y,X,L) \phi_\beta(Y,X,L)^\top\}.
    \end{align*}
\end{lemma}
In Section~\ref{sec:cgr_variance}, we prove this result and an equivalent result for IPTW (Lemma~\ref{lemma:cgr_variance_iptw}), and provide the explicit variance for the identity and logit link functions. We present the full procedure for inference on $\mu_1$ using regression adjustment in Algorithm~\ref{alg:conditional_raking} (the procedure for IPTW involves replacing steps 2--6 with their IPTW equivalents; see Algorithm~\ref{alg:cgr_iptw}). These algorithms can be implemented using standard software for marginalizing regression models [e.g., the \texttt{marginaleffects} R package \citep{marginaleffectspkg}] after estimating $\beta$ using standard software for generalized raking [e.g., the \texttt{survey} R package \citep{surveypkg}], which we will illustrate in our numerical experiments.

While $\mu_{1,n,CGR}$ is doubly-robust with respect to consistency of the missing-data model (i.e., it is consistent if either (A3) or (A4) holds), our variance estimator in Algorithm~\ref{alg:conditional_raking} is based on the influence function and is not doubly-robust; an empirical sandwich variance estimator could be used instead \citep{shook2025double}. The asymptotic behavior of $\mu_{0,n,CGR}$ is similar, and the asymptotic behavior of the ATE and RR follows from an application of the delta method. 

\begin{algorithm}
    \caption{Conditional generalized raking for the treatment-specific mean $\mu_1$ using regression adjustment}\label{alg:conditional_raking}
    \begin{algorithmic}[1]
        \State construct estimator $\pi_{n,IPCW}$ of $\pi$ using generalized linear regression;
        \State construct estimator $\eta_{\beta,n}$ of the optimal raking variable $\eta_{\beta}$;
        \State calibrate the weights: set $\pi_{n,CGR,i}^{-1} = b_i^*\pi_{n,IPCW,i}^{-1}$, where 
            \begin{align*}
                b^* =& \ \argmin_{b \in \mathbb{R}^n}\sum_{i=1}^n R_i d(b_i\pi_{n,IPCW,i}^{-1},\pi_{n,IPCW,i}^{-1}) \\
                & \text{ subject to }  \sum_{i=1}^n \eta_{\beta,n}(V_i) = \sum_{i=1}^n R_i b_i \pi_{n,IPCW,i}^{-1} \eta_{\beta,n}(V_i).
            \end{align*}
        \State obtain estimator $\beta_{n,GR}$ using generalized linear regression in the phase 2 subset weighted by $\pi_{n,CGR,i}$;
        \State construct estimator $Q_{n,1,CGR}$ of $Q$ by setting $Q_{n,1,CGR}(x,l) = h^{-1}\{(1,x,l)\beta_{n,GR}\}$;
        \State compute estimator 
        \begin{align*}
            \mu_{1,n,CGR} := \frac{1}{n}\sum_{i=1}^n \frac{R_i}{\pi_{n,CGR}(V_i)} Q_{n,1,CGR}(1,L_i);
        \end{align*}
        \State compute estimator $\sigma^2_{1,n,CGR}$ of the asymptotic variance $\sigma_{1,CGR}^2$ of $n^{1/2}(\mu_{1,n,CGR} - \mu_1)$ following Lemma~\ref{lemma:cgr_variance}: 
        \begin{align*}
            \sigma^2_{1,n,CGR} =& \ \Omega_n \Sigma_{n,\beta} \Omega_n^\top, \text{ where } \\
            \Omega_n =& \ \frac{1}{n}\sum_{i=1}^n \frac{R_i}{\pi_{n,CGR}(V_i)}\nabla h^{-1}(1,1,L_i; \beta_{n,GR}) \text{ and }\\
            \Sigma_{n,\beta} =& \ \frac{1}{n}\sum_{i=1}^n \phi_\beta(Y_i,1,L_i)\phi_\beta(Y_i,1,L_i)^\top.
        \end{align*}
    \end{algorithmic}
\end{algorithm}

The CGR estimator is appealing due to its basis in generalized raking for regression. However, there are three drawbacks to its use in practice. The first is that point estimation requires estimating nuisance parameters $\beta$ that are not of primary interest and marginalizing predictions. The second is that CGR may not be efficient, because it was not obtained by raking using the optimal raking variable for the marginal estimand \citep{breslow2009improved}. The third is that consistency for marginal parameters requires correct specification of the outcome regression model. To address these drawbacks, we could instead target the marginal estimand directly, which we pursue in the next section.

\subsection{Marginal generalized raking}\label{sec:mgr}

Our goal is now to construct estimators for marginal estimands (e.g., $\mu_1$) directly using generalized raking. As above, we first obtain an estimator $\pi_{n,IPCW}$ of $\pi$. The calibration step is the key difference between CGR and our procedure, which we call \textit{marginal generalized raking} (MGR). We will now use the EIF for the marginal target of interest to calibrate the weights. 

The estimands described in Section~\ref{sec:targets} have been well-studied in the causal inference literature, with existing non- and semi-parametric estimators. As in the case of regression parameters, the asymptotic variance of these estimators is based on their influence functions. The nonparametric EIFs for the treatment-specific means in the ideal setting (no missing data) are:
\begin{align*}
    \phi_{\mu_1,NP}(y,x,l) =& \ \frac{I(x = 1)}{P(X = 1 \mid L = l)}\{y - E(Y \mid X = x, L = l)\} + E(Y \mid X = 1, L = l) - \mu_1 \\
    \phi_{\mu_0,NP}(y,x,l) =& \ \frac{I(a = 0)}{1 - P(X = 1 \mid L = l)}\{y - E(Y \mid X = x, L = l)\} + E(Y \mid X = 0, L = l) - \mu_0,
\end{align*}
where the first subscript denotes which estimand ($\mu_1$ or $\mu_0$) $\phi$ is the EIF of, the expectations are taken under a nonparametric model, and the subscript $NP$ is shorthand for nonparametric. As above, we can define the EIF for the ATE, $\phi_{ATE,NP}(y,x,l) = \phi_{\mu_1,NP}(y,x,l) - \phi_{\mu_0,NP}(y,x,l)$. For binary outcomes, we can also define the EIF for the log-transformed RR, $\phi_{RR,NP}(y,x,l) = \phi_{\mu_1,NP}(y,x,l)/\mu_1 - \phi_{\mu_0,NP}(y,x,l)/\mu_0$.
    
If the data are known to follow a parametric model, which we have assumed here, then the nonparametric EIF can be projected into this model space, resulting in smaller variance. As above, suppose that the parametric model is indexed by parameter vector $\beta$ with score $S_\beta(y,x,l)$ and information $I_\beta$. Then the parametric EIF for the treatment-specific mean $\mu_1$ is the following function of the nonparametric EIF \citep[see, e.g., Theorem 3.5 in][]{tsiatis2006semiparametric}:
\begin{align*}
    \phi_{1,P}(y,x,l) := E\{\phi_{\mu_1,NP}(Y,X,L)S^\top_\beta(Y,X,L)\}I_\beta^{-1}S_\beta(y,x,l);
\end{align*}
the same transformation can be applied to obtain parametric EIFs for each estimand defined above. As in Section~\ref{sec:gr_regression}, in missing-data settings, the EIF for the observed-data parameter (that identifies the ideal-data parameter under the assumptions listed above) is a function of the ideal-data EIF \citep{vandervaart2000asymptotic}:
\begin{align*}
    \phi_{1,obs,P}(y,x,l,r,v) := & \ \frac{r}{\pi(v)}\phi_{1,P}(y,x,l) - \frac{\{r - \pi(v)\}}{\pi(v)}E\{\phi_{1,P}(Y,X,L)\mid V = v\}.
\end{align*}

From this construction, we can see that the optimal raking variable for estimating $\mu_1$ is now $\eta(v) = E\{\phi_{1,P}(Y,X,L) \mid V = v\}$. This corresponds to the term used in augmented inverse probability weighting for missing data \citep{robins1995semiparametric}. As described above, we can estimate $\eta$ using the phase 1 and phase 2 data in several ways, including using multiple imputation or using regression with $\phi_{1,P}$ as the outcome and $V$ as the covariates. With an estimator $\eta_n$ in hand, we can use the same procedure as in Equation \eqref{eq:gr} to calibrate the weights but with $\eta_n$ in place of $\eta_{\beta,n}$. This yields final weights $\pi_{n,MGR}$.

The MGR estimator of $\mu_1$ is
\begin{align}\label{eq:mgr}
    \mu_{1,n,MGR} :=& \ \frac{1}{n}\sum_{i=1}^n \frac{R_i}{\pi_{n,MGR}(V_i)}\left[\frac{I(X_i = 1)}{g_n(L_i)}\{Y_i - Q_n(1,L_i)\} + Q_n(1,L_i)\right],
\end{align}
where $g_n$ and $Q_n$ are estimators of the treatment propensity score and outcome regression, respectively, obtained using parametric models as described above. A similar estimator $\mu_{0,n,MGR}$ can be constructed for $\mu_0$, and these two estimators can be combined to estimate the ATE and RR. The MGR estimator in Equation~\eqref{eq:mgr} includes an inner augmented inverse probability of treatment weighting (AIPTW) estimator \citep{robins1995semiparametric}, and through calibrating the weights, is asymptotically equivalent to the augmented inverse probability of coarsening weighted (AIPCW) estimator, which is defined as 
\begin{align*}
    \mu_{1,n,AIPCW-AIPTW} :=& \ \frac{1}{n}\sum_{i=1}^n \left(\frac{R_i}{\pi_{n,IPCW}(V_i)}\left[\frac{I(X_i = 1)}{g_n(L_i)}\{Y_i - Q_n(1,L_i)\} + Q_n(1,L_i)\right] \right.\\
    & \ \ \ \ \ \ \ \ \ \ \  + \left. \frac{\{\pi_{n,IPCW} - R_i\}}{\pi_{n,IPCW}}\eta_n(V_i)\right),\ 
\end{align*} 
and is the result of using the EIF as an estimating equation.

The next theorem states this equivalence and the resulting asymptotic linearity of the MGR estimator with influence function equal to the EIF. Before stating the theorem, we define four assumptions:
\begin{itemize}
    \item [(B1)] \textit{(consistency of the outcome regression model)} $Q_n \to_P Q$;
    \item[(B2)] \textit{(consistency of the treatment-propensity score)} $g_n \to_P g$;
    \item[(B3)] \textit{(consistency of the missing-data model)} $\pi_{n,MGR} \to_P \pi$;
    \item[(B4)] \textit{(consistency of the EIF projection)} $\eta_n \to_P \eta$.
\end{itemize}
\begin{theorem}\label{thm:mgr}
    If (A1) holds and both (i) one of (B3) or (B4) and (ii) one of (B1) or (B2) hold, then the MGR estimator is asymptotically equivalent to the AIPCW-AIPTW estimator. In other words, $\mu_{1,n,MGR}$ is an asymptotically linear estimator of $\mu_1$ with influence function $\phi_{1,obs,P}$, that is,
    \begin{align*}
        \mu_{1,n,MGR} - \mu_1 = \frac{1}{n}\sum_{i=1}^n \phi_{1,obs,P}(Y_i,X_i,L_i,R_i,V_i) + o_P(n^{-1/2}).
    \end{align*}
    Because $\phi_{1,obs,P}$ is the efficient influence function, if (A1) and (B1)--(B4) hold, then $\mu_{1,n,MGR}$ is efficient.
\end{theorem}
This estimator is multiply-robust, because it is consistent if one model is correctly specified from both the missing-data level (either the missing-data model or EIF projection) and the outcome-exposure level (either the treatment-propensity score or the outcome regression model). We present the full estimation procedure in Algorithm~\ref{alg:marginal_raking}, including an influence function-based variance estimator. As in the case of $\mu_{1,n,CGR}$ above, an empirical sandwich approach could be used to obtain a doubly-robust variance estimator.

\begin{algorithm}
    \caption{Marginal generalized raking for the treatment-specific mean $\mu_1$}\label{alg:marginal_raking}
    \begin{algorithmic}[1]
        \State construct estimator $\pi_{n,IPCW}$ of $\pi$ using generalized linear regression;
        \State construct estimator $\eta_{n}$ of the optimal raking variable $\eta$
        \State calibrate the weights: set $\pi_{n,MGR,i}^{-1} = b_i^*\pi_{n,IPCW,i}^{-1}$, where 
            \begin{align*}
                b^* =& \ \argmin_{b \in \mathbb{R}^n}\sum_{i=1}^n R_i d(b_i\pi_{n,IPCW,i}^{-1},\pi_{n,IPCW,i}^{-1}) \\
                & \text{ subject to }  \sum_{i=1}^n \eta_n(V_i) = \sum_{i=1}^n R_i b_i \pi_{n,IPCW,i}^{-1} \eta_n(V_i).
            \end{align*}
        \State construct estimator $g_n$ of $g$ using, e.g., logistic regression (weighted by $\pi_{n,MGR}$ if necessary);
        \State construct estimator $Q_{n}$ of $Q$ using generalized linear regression weighted by $\pi_{n,MGR}$;
        \State compute estimator 
        \begin{align*}
            \mu_{1,n,MGR} :=& \ \frac{1}{n}\sum_{i=1}^n \frac{R_i}{\pi_{n,MGR}(V_i)}\left[\frac{I(X_i = 1)}{g_n(L_i)}\{Y_i - Q_n(1,L_i)\} + Q_n(1,L_i)\right] 
        \end{align*}
        \State compute estimator $\sigma^2_n$ of the asymptotic variance $\sigma_0^2$ of $n^{1/2}(\mu_{n,1} - \mu_1)$ using
        \begin{align*}
            \sigma^2_n := \frac{1}{n^2}\sum_{i=1}^n\sum_{j=1}^n \left[\phi_{1,obs,P,n}(Y,X,L,R,V)\phi_{1,obs,P,n}(Y,X,L,R,V)^\top\right]_{ij}.
        \end{align*}
    \end{algorithmic}
\end{algorithm}

\section{Numerical experiments}\label{sec:sims}

\subsection{Estimation procedures}

The estimators we consider are conditional generalized raking (CGR), marginal generalized raking (MGR), and AIPCW-AIPTW. For CGR, we both follow Algorithm~\ref{alg:conditional_raking} using regression adjustment to account for confounding [CGR (RA)] and use Algorithm~\ref{alg:cgr_iptw} for IPTW [CGR (IPTW)]; the latter algorithm provides a weighting-based comparator to MGR. We first estimate the probability of having full data ($\pi$) using logistic regression. Next, we estimate $g$ in the complete data, again using logistic regression, and obtain IPTW weights. To estimate $Q$ we use a linear regression model for continuous outcomes and a logistic regression model for binary outcomes. Then we use multiple imputation \citep[MI; ][]{vanbuuren2018flexible} with 10 imputed datasets, implemented in the R package \texttt{mice} \citep{micepkg}, to obtain the estimated optimal raking variables $\eta_{\beta,n}$; we calibrate the initial weights and obtain a final estimator using the R package \texttt{survey} \citep{surveypkg}. Finally, we use the \texttt{marginaleffects} R package \citep{marginaleffectspkg} to marginalize and obtain delta method-based confidence intervals (equivalent to Algorithm~\ref{alg:conditional_raking}). For MGR, we follow Algorithm~\ref{alg:marginal_raking}. We estimate the probability of having full data and the treatment-assignment propensity score using logistic regression. We estimate the outcome regression using generalized linear regression (linear regression for continuous outcomes and logistic regression for binary outcomes). As in CGR, we use MI to estimate the optimal raking variables, now based on the marginal EIFs, and calibrate the weights using \texttt{survey}. Finally, we computed an AIPCW-AIPTW estimator as a comparator, using the same approaches as MGR for the missing-data model, treatment-assignment propensity score, outcome regression model, and optimal raking variable. 

\subsection{Data-generating scenarios}

We consider five scenarios with missing data, and for each scenario generate a continuous outcome and binary outcome. The scenarios are designed to show the performance of the methods in correctly specified scenarios (Scenarios 1 and 2) and scenarios with misspecification of the treatment propensity score (Scenario 3), outcome regression model (Scenario 4), and missing-data model (Scenario 5). This misspecification is achieved by omitting variables (denoted by $U$ below) from the fitted models that are part of the true data-generating model. 

We generate data following a similar structure to \citet{williamson2026assessing}. We first generate covariates $L=(Z,W,U,A)$, where we have two always-observed confounders $Z = (Z_1, Z_2)$, two confounders subject to missingness $W = (W_1, W_2)$, and four variables $U = (U_1,U_2)$ and $C = (C_1, C_2)$ that are confounders only in certain scenarios. $L$ follows a multivariate normal distribution with mean zero and pairwise correlation 0 between variables (purely for simplicity of the illustration; our methodology does not require uncorrelated variables). Next, we create a binary exposure variable $X \sim Bern\{\expit([1,L]\beta_X)\},$ where $\beta_X$ varies based on the scenario. Based on $L$ and $X$, we generate outcomes and missing data in $W$. The outcome in all cases is generated according to a generalized linear model, with 
\begin{align*}
    h\{E(Y \mid X = x, Z = z, W = w, U = u, C = c)\} = (1,x,z,w,u,c)\beta_0,
\end{align*}
where $\beta_0$ depends on the scenario (Table~\ref{tab:sim_model_coefs}). In continuous outcome settings (sub-scenario ``a''), $h$ is the identity and we generate Gaussian errors with mean zero and variance 1. The true regression coefficients depend on the scenario. In binary outcome settings, $h$ is the logit function, yielding a Bernoulli outcome with marginal probability approximately 25\% (Scenario 1b) or 10\% (Scenarios 2b, 3b, 4b, 5b). We generate missing data following a Bernoulli distribution with mean $P(R = 0 \mid X = x, Y = y, Z = z, C = c, U = u) = (1,x,y,z,c,u)\alpha_0,$ with $\alpha_0$ for each scenario defined in Table~\ref{tab:sim_model_coefs}, leading to marginal probability of missing data approximately 20\% (Scenario 1) or 50\% (Scenarios 2--5). The true ATE and RR for each scenario are provided in Table~\ref{tab:true_values}. In Scenario 1, $U$ and $C$ are not confounders. In Scenario 2, $U$ are not confounders but $C$ are. In Scenarios 3 and 4, both $U$ and $C$ are confounders; in Scenario 5, $U$ are part of the missing-data model and $C$ are confounders.

For each scenario, we generated 2500 independent datasets following the specification described above. For continuous outcomes (sub-scenario ``a'') we generated data with sample size $n \in \{250, 500, 1000, 2000\}$. For binary outcomes (sub-scenario ``b'') we generated data with sample size $n \in \{1000, 2000, 4000, 8000\}$, which resulted in similar effective sample size (number of events) to the continuous outcome cases in Scenario 1b and slightly smaller effective sample size in Scenarios 2b--5b. On each dataset, we fit CGR-RA, CGR-IPTW, MGR, and AIPCW-AIPTW, obtaining a point estimate and 95\% confidence interval. In Scenarios 1 and 2, we fit correctly-specified regression models for the missing-data model, treatment propensity score model, and outcome regression model. In Scenario 3, the fitted treatment propensity score model was misspecified by omitting $U$, but the fitted outcome regression model included $U$. This reflects a case where the variables $U$ are not known confounders but are included just in the outcome model. In Scenario 4, the fitted treatment propensity score model included $U$ but the fitted outcome regression model was misspecified by omitting $U$. In Scenario 5, the fitted missing data model was misspecified by omitting $U$ but all other models were correctly specified. These latter four scenarios were designed to illustrate the performance of the estimators under misspecification. Importantly, due to the structure of the data-generating mechanism, in Scenario 4 both the outcome regression model and the optimal raking variable were misspecified, so we do not expect any method to be fully efficient. 

We assessed the performance of each estimator with respect to bias (median across the 2500 replications), median asymptotic standard error (ASE), empirical standard error, median absolute deviation (MAD; also computed on the same scale as the ASE), robust root mean squared error (rRMSE; $\sqrt{\text{median bias}^2 + \text{MAD}^2}$), relative efficiency with respect to robust MSE (comparing rRMSE to MGR; values greater than 1 imply more efficient) and confidence interval coverage based on asymptotic or empirical standard errors. We compute ESE and MAD on the log scale for the RR because the asymptotic standard error for estimators of the RR are computed on this scale.

\begin{table}
\centering
\caption{True values of estimands in the simulation scenarios. Scenarios ending in ``a'' involve a continuous outcome, so the $RR$ is not defined.\label{tab:true_values}}
\centering
\fontsize{11}{13}\selectfont
\begin{tabular}[t]{lrrrrrrrrrr}
\toprule
\multicolumn{1}{c}{ } & \multicolumn{10}{c}{Scenario} \\
\cmidrule(l{3pt}r{3pt}){2-11}
Estimand & 1a & 1b & 2a & 2b & 3a & 3b & 4a & 4b & 5a & 5b\\
\midrule
$\mu_0$ & -0.25 & 0.16 & -4.5 & 0.06 & -9.5 & 0.07 & -9.5 & 0.07 & -4.5 & 0.06\\
$\mu_1$ & 1.25 & 0.34 & -3 & 0.15 & -8 & 0.11 & -8 & 0.11 & -3 & 0.15\\
$ATE$ & 1.5 & 0.17 & 1.5 & 0.09 & 1.5 & 0.04 & 1.5 & 0.04 & 1.5 & 0.09\\
$RR$ &  & 2.07 &  & 2.48 &  & 1.52 &  & 1.52 &  & 2.48\\
\bottomrule
\end{tabular}
\end{table}

\subsection{Results with correctly-specified models}

We first present results with correctly-specified models (Scenarios 1 and 2). Results for our base-case scenario (Scenario 1, Tables~\ref{tab:sim_combined_performance_tab_s1_ATE}--\ref{tab:sim_combined_performance_tab_s2_RR}) with a common outcome and approximately 20\% missing data were similar to Scenario 2 with approximately 50\% missing data, so we focus on Scenario 2. Results for the continuous outcome (Scenario 2a) are in Table~\ref{tab:sim_combined_performance_tab_s3_ATE}. Results for the 10\% marginal probability binary outcome (2b) are in Tables~\ref{tab:sim_combined_performance_tab_s4_ATE} (ATE) and \ref{tab:sim_combined_performance_tab_s4_RR} (RR). We focus on the binary-outcome scenario to discuss both the ATE and RR. At all sample sizes and for both the ATE and RR, MGR has similar bias and variance, and therefore relative efficiency, to AIPCW-AIPTW. While the two estimators are asymptotically equivalent, so we would expect to see similar performance in large sample sizes, it is interesting that they are similar even in small samples. Relative efficiency when estimating the ATE and RR using CGR (RA) is also similar to MGR in this scenario, particularly at large sample sizes; in contrast, relative efficiency using CGR (IPTW) is lower than MGR (Table~\ref{tab:sim_combined_performance_tab_s4_cgr_iptw}). Coverage based on the ASE is within Monte-Carlo error of the nominal level at large sample size, while coverage based on the ESE is always within Monte-Carlo error of the nominal level. This implies that our variance estimators are consistent. The lower relative efficiency of CGR (IPTW) matches with theory, because the raking variable used with CGR (IPTW) is not the optimal raking variable. We did not expect the efficiency of CGR (RA) to be so close to the efficiency of MGR, because these estimators are consistent in different model classes \citep{lumley2011connections}. However, the EIF for the regression parameter for $X$ in CGR is highly correlated with the marginal EIFs for the RR and ATE, and the efficiency gain for GR is tied to the correlation with the EIF \citep{lumley2011connections}. 

\begin{table}[!h]
\centering
\caption{Performance of estimators of the ATE in Scenario 2b, a binary outcome setting with approximately 50\% missing data and correctly specified nuisance models.\label{tab:sim_combined_performance_tab_s4_ATE}}
\centering
\fontsize{9}{11}\selectfont
\begin{threeparttable}
\begin{tabular}[t]{>{\raggedright\arraybackslash}p{3.5em}lrrrrrrrrr}
\toprule
Est. & n & Med. bias & ASE & ESE & MAD & rRMSE & RMSE & Rel. Eff. & Cover & ESE Cover\\
\midrule
AIPCW-AIPTW &  & $<$ 0.0001 & 0.021 & 0.021 & 0.020 & 0.020 & 0.021 & 0.932 & 0.956 & 0.952\\
\cmidrule{1-1}
\cmidrule{3-11}
CGR (RA) &  & -4 $\times 10^{-04}$ & 0.019 & 0.019 & 0.019 & 0.019 & 0.019 & 1.023 & 0.937 & 0.946\\
\cmidrule{1-1}
\cmidrule{3-11}
MGR & \multirow{-3}{*}[2\dimexpr\aboverulesep+\belowrulesep+\cmidrulewidth]{\raggedright\arraybackslash 1000} & -7 $\times 10^{-04}$ & 0.021 & 0.020 & 0.019 & 0.019 & 0.020 & 1.000 & 0.960 & 0.950\\
\cmidrule{1-11}
AIPCW-AIPTW &  & 9 $\times 10^{-04}$ & 0.014 & 0.014 & 0.014 & 0.014 & 0.014 & 0.968 & 0.938 & 0.945\\
\cmidrule{1-1}
\cmidrule{3-11}
CGR (RA) &  & 2 $\times 10^{-04}$ & 0.013 & 0.014 & 0.013 & 0.013 & 0.014 & 1.017 & 0.940 & 0.948\\
\cmidrule{1-1}
\cmidrule{3-11}
MGR & \multirow{-3}{*}[2\dimexpr\aboverulesep+\belowrulesep+\cmidrulewidth]{\raggedright\arraybackslash 2000} & 4 $\times 10^{-04}$ & 0.014 & 0.014 & 0.014 & 0.014 & 0.014 & 1.000 & 0.949 & 0.946\\
\cmidrule{1-11}
AIPCW-AIPTW &  & 3 $\times 10^{-04}$ & 0.010 & 0.010 & 0.010 & 0.010 & 0.010 & 0.946 & 0.944 & 0.951\\
\cmidrule{1-1}
\cmidrule{3-11}
CGR (RA) &  & -5 $\times 10^{-04}$ & 0.009 & 0.009 & 0.010 & 0.010 & 0.009 & 0.999 & 0.950 & 0.954\\
\cmidrule{1-1}
\cmidrule{3-11}
MGR & \multirow{-3}{*}[2\dimexpr\aboverulesep+\belowrulesep+\cmidrulewidth]{\raggedright\arraybackslash 4000} & -5 $\times 10^{-04}$ & 0.010 & 0.010 & 0.010 & 0.010 & 0.010 & 1.000 & 0.950 & 0.953\\
\cmidrule{1-11}
AIPCW-AIPTW &  & 6 $\times 10^{-04}$ & 0.007 & 0.007 & 0.007 & 0.007 & 0.007 & 0.959 & 0.932 & 0.948\\
\cmidrule{1-1}
\cmidrule{3-11}
CGR (RA) &  & -2 $\times 10^{-04}$ & 0.007 & 0.007 & 0.007 & 0.007 & 0.007 & 1.000 & 0.939 & 0.946\\
\cmidrule{1-1}
\cmidrule{3-11}
MGR & \multirow{-3}{*}[2\dimexpr\aboverulesep+\belowrulesep+\cmidrulewidth]{\raggedright\arraybackslash 8000} & -3 $\times 10^{-04}$ & 0.007 & 0.007 & 0.007 & 0.007 & 0.007 & 1.000 & 0.941 & 0.948\\
\bottomrule
\end{tabular}
\begin{tablenotes}
\item Abbreviations: ATE: average treatment effect; RR: relative risk; RR: relative risk; Est: Estimator; Med: Median; ASE: median asymptotic standard error; ESE: empirical standard error; MAD: mean absolute deviation; rRMSE: robust root mean squared error (RMSE) (square root of med bias squared + MAD squared); RMSE: RMSE (mean bias squared + ESE squared); Rel. Eff.: relative efficiency (rRMSE relative to that of MGR); Cover: coverage based on the ASE; ESE Cover: coverage based on the ESE; CGR: conditional generalized raking (GR); IPTW: inverse probability of treatment weighting; RA: regression adjustment; MGR: marginal GR; AIPCW: augmented inverse probability of coarsening weighting; AIPTW: augmented IPTW. Results are based on 2500 Monte-Carlo replications. 
\end{tablenotes}
\end{threeparttable}
\end{table}

\begin{table}[!h]
\centering
\caption{Performance of estimators of the RR in Scenario 2b, a binary outcome setting with approximately 50\% missing data and correctly specified nuisance models.\label{tab:sim_combined_performance_tab_s4_RR}}
\centering
\fontsize{9}{11}\selectfont
\begin{threeparttable}
\begin{tabular}[t]{>{\raggedright\arraybackslash}p{3.5em}lrrrrrrrrr}
\toprule
Est. & n & Med. bias & ASE & ESE & MAD & rRMSE & RMSE & Rel. Eff. & Cover & ESE Cover\\
\midrule
AIPCW-AIPTW &  & -0.0200 & 0.230 & 0.212 & 0.205 & 0.206 & 0.217 & 0.956 & 0.964 & 0.948\\
\cmidrule{1-1}
\cmidrule{3-11}
CGR (RA) &  & -0.0067 & 0.196 & 0.202 & 0.194 & 0.195 & 0.209 & 1.012 & 0.940 & 0.948\\
\cmidrule{1-1}
\cmidrule{3-11}
MGR & \multirow{-3}{*}[2\dimexpr\aboverulesep+\belowrulesep+\cmidrulewidth]{\raggedright\arraybackslash 1000} & 0.0053 & 0.231 & 0.205 & 0.197 & 0.197 & 0.211 & 1.000 & 0.972 & 0.951\\
\cmidrule{1-11}
AIPCW-AIPTW &  & -0.0077 & 0.149 & 0.144 & 0.144 & 0.144 & 0.147 & 0.980 & 0.959 & 0.946\\
\cmidrule{1-1}
\cmidrule{3-11}
CGR (RA) &  & 0.0066 & 0.137 & 0.139 & 0.139 & 0.139 & 0.145 & 1.014 & 0.948 & 0.947\\
\cmidrule{1-1}
\cmidrule{3-11}
MGR & \multirow{-3}{*}[2\dimexpr\aboverulesep+\belowrulesep+\cmidrulewidth]{\raggedright\arraybackslash 2000} & 0.0122 & 0.150 & 0.139 & 0.140 & 0.141 & 0.145 & 1.000 & 0.963 & 0.947\\
\cmidrule{1-11}
AIPCW-AIPTW &  & -0.0185 & 0.101 & 0.099 & 0.098 & 0.100 & 0.100 & 0.966 & 0.952 & 0.948\\
\cmidrule{1-1}
\cmidrule{3-11}
CGR (RA) &  & -0.0069 & 0.096 & 0.096 & 0.096 & 0.096 & 0.097 & 1.008 & 0.950 & 0.951\\
\cmidrule{1-1}
\cmidrule{3-11}
MGR & \multirow{-3}{*}[2\dimexpr\aboverulesep+\belowrulesep+\cmidrulewidth]{\raggedright\arraybackslash 4000} & -0.0098 & 0.101 & 0.097 & 0.096 & 0.097 & 0.098 & 1.000 & 0.956 & 0.949\\
\cmidrule{1-11}
AIPCW-AIPTW &  & -0.0185 & 0.070 & 0.070 & 0.070 & 0.072 & 0.070 & 0.940 & 0.945 & 0.945\\
\cmidrule{1-1}
\cmidrule{3-11}
CGR (RA) &  & -0.0053 & 0.067 & 0.068 & 0.066 & 0.066 & 0.068 & 1.019 & 0.940 & 0.945\\
\cmidrule{1-1}
\cmidrule{3-11}
MGR & \multirow{-3}{*}[2\dimexpr\aboverulesep+\belowrulesep+\cmidrulewidth]{\raggedright\arraybackslash 8000} & -0.0068 & 0.070 & 0.069 & 0.067 & 0.068 & 0.069 & 1.000 & 0.952 & 0.946\\
\bottomrule
\end{tabular}
\begin{tablenotes}
\item Abbreviations: ATE: average treatment effect; RR: relative risk; RR: relative risk; Est: Estimator; Med: Median; ASE: median asymptotic standard error; ESE: empirical standard error; MAD: mean absolute deviation; rRMSE: robust root mean squared error (RMSE) (square root of med bias squared + MAD squared); RMSE: RMSE (mean bias squared + ESE squared); Rel. Eff.: relative efficiency (rRMSE relative to that of MGR); Cover: coverage based on the ASE; ESE Cover: coverage based on the ESE; CGR: conditional generalized raking (GR); IPTW: inverse probability of treatment weighting; RA: regression adjustment; MGR: marginal GR; AIPCW: augmented inverse probability of coarsening weighting; AIPTW: augmented IPTW. Results are based on 2500 Monte-Carlo replications. 
\end{tablenotes}
\end{threeparttable}
\end{table}

\subsection{Results with misspecified models}

We now consider scenarios with misspecified models for either the treatment propensity score (Scenario 3), outcome regression and optimal raking variable (Scenario 4), or missing-data model (Scenario 5). As above, we focus on the binary-outcome results to investigate both the ATE and RR. Continuous-outcome results are similar at a lower overall sample size, which we expected. In Scenario 3, the fitted propensity score model was misspecified by omitting two confounders. This resulted in substantial bias for CGR (IPTW) and a corresponding loss of efficiency (expected due to Lemma~\ref{lemma:cgr_variance_iptw}), while all other estimators remained unbiased (Tables~\ref{tab:sim_combined_performance_tab_s6_ATE}, \ref{tab:sim_combined_performance_tab_s5_ATE}--\ref{tab:sim_combined_performance_tab_s6_cgr_iptw}). CGR (RA) was more efficient than MGR or AIPCW-AIPTW in this setting, as expected from theory; the latter two estimators incorporate the misspecified propensity score model, so double-robustness with respect to $Q$ and $g$ guarantees consistency but not efficiency. In Scenario 4, the fitted outcome regression model and the optimal raking variables were misspecified by omitting the two confounders. This resulted in substantial bias for CGR (RA) and a corresponding loss of efficiency, while all other estimators remained unbiased (Tables~\ref{tab:sim_combined_performance_tab_s8_ATE}, \ref{tab:sim_combined_performance_tab_s7_ATE}--\ref{tab:sim_combined_performance_tab_s8_cgr_iptw}). CGR (IPTW) was the most efficient in this case. Again, both results are expected from theory, because MGR and AIPCW-AIPTW include the misspecified outcome regression model. CGR (RA) requires a correctly-specified estimator of $Q$ (Lemma~\ref{lemma:cgr_variance}), but in this case CGR (IPTW) is consistent because the outcome model is based only on the treatment variable and confounding adjustment relies on the correctly-specific treatment propensity model. 
Finally, in Scenario 5, the fitted missing-data model was misspecified. In this scenario, all methods were consistent and had small variance, resulting in efficiency all within Monte-Carlo error of each other (Tables~\ref{tab:sim_combined_performance_tab_s9_ATE}--\ref{tab:sim_combined_performance_tab_s10_RR}), as expected due to Theorem~\ref{thm:mgr} and the consistency results from Table 1 of \citet{williamson2026assessing}. 

\begin{table}[!h]
\centering
\caption{Performance of estimators of the ATE in Scenario 3b, a binary outcome setting with approximately 50\% missing data and a misspecified treatment propensity model.\label{tab:sim_combined_performance_tab_s6_ATE}}
\centering
\fontsize{9}{11}\selectfont
\begin{threeparttable}
\begin{tabular}[t]{>{\raggedright\arraybackslash}p{3.5em}lrrrrrrrrr}
\toprule
Est. & n & Med. bias & ASE & ESE & MAD & rRMSE & RMSE & Rel. Eff. & Cover & ESE Cover\\
\midrule
AIPCW-AIPTW &  & 2 $\times 10^{-04}$ & 0.017 & 0.017 & 0.016 & 0.016 & 0.017 & 1.001 & 0.945 & 0.946\\
\cmidrule{1-1}
\cmidrule{3-11}
CGR (RA) &  & 4 $\times 10^{-04}$ & 0.014 & 0.015 & 0.015 & 0.015 & 0.015 & 1.118 & 0.916 & 0.952\\
\cmidrule{1-1}
\cmidrule{3-11}
MGR & \multirow{-3}{*}[2\dimexpr\aboverulesep+\belowrulesep+\cmidrulewidth]{\raggedright\arraybackslash 1000} & 4 $\times 10^{-04}$ & 0.017 & 0.016 & 0.016 & 0.016 & 0.016 & 1.000 & 0.951 & 0.948\\
\cmidrule{1-11}
AIPCW-AIPTW &  & 5 $\times 10^{-04}$ & 0.010 & 0.011 & 0.011 & 0.011 & 0.011 & 0.961 & 0.928 & 0.949\\
\cmidrule{1-1}
\cmidrule{3-11}
CGR (RA) &  & $<$ 0.0001 & 0.009 & 0.010 & 0.010 & 0.010 & 0.010 & 1.090 & 0.942 & 0.946\\
\cmidrule{1-1}
\cmidrule{3-11}
MGR & \multirow{-3}{*}[2\dimexpr\aboverulesep+\belowrulesep+\cmidrulewidth]{\raggedright\arraybackslash 2000} & 2 $\times 10^{-04}$ & 0.010 & 0.011 & 0.011 & 0.011 & 0.011 & 1.000 & 0.936 & 0.948\\
\cmidrule{1-11}
AIPCW-AIPTW &  & 8 $\times 10^{-04}$ & 0.007 & 0.008 & 0.007 & 0.007 & 0.008 & 0.997 & 0.900 & 0.941\\
\cmidrule{1-1}
\cmidrule{3-11}
CGR (RA) &  & 3 $\times 10^{-04}$ & 0.007 & 0.007 & 0.007 & 0.007 & 0.007 & 1.079 & 0.944 & 0.949\\
\cmidrule{1-1}
\cmidrule{3-11}
MGR & \multirow{-3}{*}[2\dimexpr\aboverulesep+\belowrulesep+\cmidrulewidth]{\raggedright\arraybackslash 4000} & 4 $\times 10^{-04}$ & 0.007 & 0.008 & 0.007 & 0.007 & 0.008 & 1.000 & 0.914 & 0.947\\
\cmidrule{1-11}
AIPCW-AIPTW &  & 7 $\times 10^{-04}$ & 0.004 & 0.005 & 0.005 & 0.005 & 0.005 & 0.967 & 0.883 & 0.943\\
\cmidrule{1-1}
\cmidrule{3-11}
CGR (RA) &  & 2 $\times 10^{-04}$ & 0.005 & 0.005 & 0.005 & 0.005 & 0.005 & 1.097 & 0.943 & 0.951\\
\cmidrule{1-1}
\cmidrule{3-11}
MGR & \multirow{-3}{*}[2\dimexpr\aboverulesep+\belowrulesep+\cmidrulewidth]{\raggedright\arraybackslash 8000} & 5 $\times 10^{-04}$ & 0.004 & 0.005 & 0.005 & 0.005 & 0.005 & 1.000 & 0.892 & 0.949\\
\bottomrule
\end{tabular}
\begin{tablenotes}
\item Abbreviations: ATE: average treatment effect; RR: relative risk; RR: relative risk; RR: relative risk; Est: Estimator; Med: Median; ASE: median asymptotic standard error; ESE: empirical standard error; MAD: mean absolute deviation; rRMSE: robust root mean squared error (RMSE) (square root of med bias squared + MAD squared); RMSE: RMSE (mean bias squared + ESE squared); Rel. Eff.: relative efficiency (rRMSE relative to that of MGR); Cover: coverage based on the ASE; ESE Cover: coverage based on the ESE; CGR: conditional generalized raking (GR); IPTW: inverse probability of treatment weighting; RA: regression adjustment; MGR: marginal GR; AIPCW: augmented inverse probability of coarsening weighting; AIPTW: augmented IPTW. Results are based on 2500 Monte-Carlo replications. 
\end{tablenotes}
\end{threeparttable}
\end{table}

\begin{table}[!h]
\centering
\caption{Performance of estimators of the ATE in Scenario 4b, a binary outcome setting with approximately 50\% missing data and a misspecified outcome regression model.\label{tab:sim_combined_performance_tab_s8_ATE}}
\centering
\fontsize{9}{11}\selectfont
\begin{threeparttable}
\begin{tabular}[t]{>{\raggedright\arraybackslash}p{3.5em}lrrrrrrrrr}
\toprule
Est. & n & Med. bias & ASE & ESE & MAD & rRMSE & RMSE & Rel. Eff. & Cover & ESE Cover\\
\midrule
AIPCW-AIPTW &  & 0.0016 & 0.022 & 0.034 & 0.029 & 0.030 & 0.034 & 0.922 & 0.876 & 0.956\\
\cmidrule{1-1}
\cmidrule{3-11}
CGR (RA) &  & 0.0386 & 0.020 & 0.020 & 0.019 & 0.043 & 0.043 & 0.631 & 0.515 & 0.492\\
\cmidrule{1-1}
\cmidrule{3-11}
MGR & \multirow{-3}{*}[2\dimexpr\aboverulesep+\belowrulesep+\cmidrulewidth]{\raggedright\arraybackslash 1000} & 0.0016 & 0.022 & 0.033 & 0.027 & 0.027 & 0.033 & 1.000 & 0.891 & 0.956\\
\cmidrule{1-11}
AIPCW-AIPTW &  & 0.0014 & 0.015 & 0.023 & 0.021 & 0.021 & 0.023 & 0.979 & 0.830 & 0.950\\
\cmidrule{1-1}
\cmidrule{3-11}
CGR (RA) &  & 0.0380 & 0.014 & 0.014 & 0.013 & 0.040 & 0.040 & 0.503 & 0.202 & 0.204\\
\cmidrule{1-1}
\cmidrule{3-11}
MGR & \multirow{-3}{*}[2\dimexpr\aboverulesep+\belowrulesep+\cmidrulewidth]{\raggedright\arraybackslash 2000} & 0.0011 & 0.015 & 0.022 & 0.020 & 0.020 & 0.022 & 1.000 & 0.844 & 0.954\\
\cmidrule{1-11}
AIPCW-AIPTW &  & 0.0009 & 0.010 & 0.016 & 0.015 & 0.015 & 0.016 & 0.987 & 0.800 & 0.951\\
\cmidrule{1-1}
\cmidrule{3-11}
CGR (RA) &  & 0.0382 & 0.010 & 0.010 & 0.009 & 0.039 & 0.039 & 0.379 & 0.023 & 0.025\\
\cmidrule{1-1}
\cmidrule{3-11}
MGR & \multirow{-3}{*}[2\dimexpr\aboverulesep+\belowrulesep+\cmidrulewidth]{\raggedright\arraybackslash 4000} & 0.0008 & 0.010 & 0.016 & 0.015 & 0.015 & 0.016 & 1.000 & 0.804 & 0.951\\
\cmidrule{1-11}
AIPCW-AIPTW &  & 0.0010 & 0.007 & 0.012 & 0.011 & 0.011 & 0.012 & 0.989 & 0.772 & 0.957\\
\cmidrule{1-1}
\cmidrule{3-11}
CGR (RA) &  & 0.0377 & 0.007 & 0.007 & 0.007 & 0.038 & 0.038 & 0.280 & 0.000 & 0.000\\
\cmidrule{1-1}
\cmidrule{3-11}
MGR & \multirow{-3}{*}[2\dimexpr\aboverulesep+\belowrulesep+\cmidrulewidth]{\raggedright\arraybackslash 8000} & 0.0005 & 0.007 & 0.012 & 0.011 & 0.011 & 0.012 & 1.000 & 0.780 & 0.955\\
\bottomrule
\end{tabular}
\begin{tablenotes}
\item Abbreviations: ATE: average treatment effect; RR: relative risk; RR: relative risk; RR: relative risk; RR: relative risk; Est: Estimator; Med: Median; ASE: median asymptotic standard error; ESE: empirical standard error; MAD: mean absolute deviation; rRMSE: robust root mean squared error (RMSE) (square root of med bias squared + MAD squared); RMSE: RMSE (mean bias squared + ESE squared); Rel. Eff.: relative efficiency (rRMSE relative to that of MGR); Cover: coverage based on the ASE; ESE Cover: coverage based on the ESE; CGR: conditional generalized raking (GR); IPTW: inverse probability of treatment weighting; RA: regression adjustment; MGR: marginal GR; AIPCW: augmented inverse probability of coarsening weighting; AIPTW: augmented IPTW. Results are based on 2500 Monte-Carlo replications. 
\end{tablenotes}
\end{threeparttable}
\end{table}

\section{Analysis of the VCCC data}\label{sec:data_analysis}

We now analyze data from a cohort of people living with HIV (PWH; we will also use this abbreviation refer to a person living with HIV) who received care at the Vanderbilt Comprehensive Care Clinic (VCCC) in Nashville, Tennessee \citep{oh2018considerations,giganti2020accounting}. The VCCC electronic health record (EHR) contains routinely-collected data related to clinical care, including demographics, laboratory measurements, pharmacy dispensations, clinical events, and vital status. There are both validated and unvalidated records for 4797 PWH who received care at the VCCC between 1998 and 2013.

Our goal is to estimate the average treatment effect and relative risk of a PWH experiencing an AIDS-defining event (ADE) or death within 5 years of treatment initiation, comparing PWH who initiated a combination antiretroviral therapy (ART) regimen containing a protease inhibitor (PI) against those who initiated a combination ART regimen that did not contain a PI. To ensure a clear illustration, we applied the following exclusion criteria to both the unvalidated and validated datasets and selected PWH who met the criteria from both datasets to ensure that any differences between estimators are not due to differences in PWH included in the sample \citep{oh2021improved}. We excluded PWH if they had a missing ART start date, started ART prior to enrollment, started an ART regimen with fewer than 3 medications, had an ADE before starting ART, or had fewer than 3.5 years of follow-up time after starting ART. In the unvalidated data, 824 PWH remained after applying these exclusion criteria; 847 remained after applying the criteria in the validated dataset. The final cohort of PWH who passed exclusion criteria in both data sources had 628 individuals.

Because we have validated data on all participants in the final sample, we can compute the ATE and RR in this sample and compare results against these values. To mimic a two-phase study where we did not have validated data for all PWH, we retrospectively sampled 1000 different validation subsets using a case-control sampling design, where we sampled all error-prone cases, i.e., all people with a possibly error-prone ADE or death within 5 years indicator, and a simple random sample of error-prone controls. We assumed that we had the budget to validate 314 records (50\%). Because we fixed the phase 1 cohort of 628 individuals, and thus only capture phase 2 variance, we do not evaluate the ASE (which estimates the total variance from phase 1 and phase 2). Instead, we only compare the ESE for all approaches. 

We next specify our validation, treatment propensity, and outcome regression models. We modeled the probability of being validated as a function of the error-prone ADE or death indicator, the error-prone indicator of starting on a combination ART regimen containing a PI, the error-prone (square-root transformed) CD4 count, error-prone year of ART initiation, age, an indicator of Black race, and an indicator of injection drug use (IDU). 
% We modeled the propensity of starting on a combination ART regimen containing a PI as a function of the error-prone PI indicator and the error-free versions of CD4 count, year of ART initiation, age, Black race, and IDU. 
We modeled the propensity of starting on a combination ART regimen containing a PI as a function of the error-free versions of CD4 count, year of ART initiation, age, Black race, and IDU. 
Our outcome model included the indicator of being on a PI-containing regimen and was adjusted for CD4 count, year of ART initiation, age, Black race, and IDU. Because we have error-prone versions of variables on all PWH, $V$ (the phase 1 variables) consisted of sex, age, Black race, IDU, and error-prone versions of the ADE or death indicator, indicator of starting on a combination ART regimen containing a PI, CD4 count, and year of ART initiation. We used MI with 10 imputed datasets to estimate the optimal raking variables. As in the simulations, we use Algorithm~\ref{alg:conditional_raking} to obtain the CGR (RA) estimator, Algorithm~\ref{alg:cgr_iptw} to obtain the CGR (IPTW) estimator, and Algorithm~\ref{alg:marginal_raking} to obtain the MGR estimator.

Of the 628 PWH, in the validated data 61 experienced an event, while 112 experienced the event in the unvalidated data; 57 outcomes were misclassified (positive predictive value [PPV] 51.8\%; sensitivity 0.951; specificity 0.905; negative predictive value 0.994; Table~\ref{tab:vccc_unvalidated_class_perf}). Similarly, 266 started a PI-containing regimen in the validated cohort, while in the unvalidated data, 267 started a PI-containing regimen; 9 records were misclassified (PPV 98.1\%). There was also measurement error in year of treatment initiation and CD4 count at treatment initiation. Overall, 97 records (15\%) had an error in any of these variables (Table~\ref{tab:vccc_num_prop_errors}). The estimated probability of being observed was bounded away from 0 (minimum value 0.32), and the estimated probability of treatment was bounded away from both 0 and 1 in both the validated and unvalidated data (ranging from 0.24 to 0.63). For further description of the cohort, see Table~\ref{tab:vccc_descriptive}.  

We present the results from our analysis in Table~\ref{tab:vccc_est}. The true values, estimated in the full validated data, are 0.05 (ATE) and 1.676 (RR) when using a g-formula or AIPTW to obtain these values; using IPTW results in 0.06 (ATE) and 1.89 (RR). We use the former values as the truth because two methods agreed. We observed nearly unbiased estimation of the ATE and RR using MGR and CGR (RA). In contrast, there was more bias for estimating both the ATE and RR using CGR (IPTW), which is close to the 0.06 (ATE) and 1.89 (RR) values instead. The ESE was similar between all three methods (all were within Monte-Carlo error of each other). The large bias for the CGR (IPTW) method led to a large rRMSE, leading to a substantial loss of efficiency relative to MGR, with CGR (IPTW) 30\% less efficient than MGR for estimating the ATE and nearly 50\% less efficient for the RR. This bias suggests that the propensity score model is difficult to estimate in this scenario. While CGR (RA) was approximately 10\% less efficient than MGR for estimating both the ATE and the RR, this was within Monte-Carlo error in this example. Finally, as expected due to the degree of measurement error in this cohort, the naive analysis on just the unvalidated data results in an estimated ATE of near zero (5$\times 10^{-4}$) and a relative risk of 0.997, which is substantially different than the estimates from the completely-validated data. 

\begin{table}[!h]
\centering
\caption{Performance of estimators of the effect of starting a protease inhibitor (PI)-containing combination antiretroviral therapy (ART) regimen on experiencing an AIDS-defining event or death within 5 years of ART initiation compared to an non-PI-containing regimen,  in the VCCC data. \label{tab:vccc_est}}
\centering
\fontsize{9}{11}\selectfont
\begin{threeparttable}
\begin{tabular}[t]{llrrrrrr}
\toprule
Estimand & Procedure & Med. est. & Med. bias & ESE & MAD & rRMSE & Rel. Eff.\\
\midrule
 & CGR (IPTW) & 0.062 & 0.012 & 0.011 & 0.010 & 0.016 & 0.695\\
\cmidrule{2-8}
 & CGR (RA) & 0.051 & 0.002 & 0.012 & 0.012 & 0.012 & 0.898\\
\cmidrule{2-8}
\multirow{-3}{*}{\raggedright\arraybackslash ATE} & MGR & 0.051 & 0.002 & 0.011 & 0.011 & 0.011 & 1.000\\
\cmidrule{1-8}
 & CGR (IPTW) & 1.882 & 0.206 & 0.109 & 0.110 & 0.233 & 0.501\\
\cmidrule{2-8}
 & CGR (RA) & 1.695 & 0.019 & 0.125 & 0.131 & 0.133 & 0.881\\
\cmidrule{2-8}
\multirow{-3}{*}{\raggedright\arraybackslash RR} & MGR & 1.696 & 0.019 & 0.113 & 0.115 & 0.117 & 1.000\\
\bottomrule
\end{tabular}
\begin{tablenotes}
\item Abbreviations: ATE: average treatment effect; RR: relative risk; CGR: conditional generalized raking (GR); IPTW: inverse probability of treatment weighting; RA: regression adjustment; MGR: marginal GR; Med.: median; est.: estimate; ESE: empirical standard error; MAD: median absolute deviation; rRMSE: robust root mean squared error (square root of median bias squared + MAD squared); Rel. Eff.: relative efficiency (rRMSE relative to that of MGR). Results are based on 1000 retrospectively-sampled datasets. 
\end{tablenotes}
\end{threeparttable}
\end{table}

\section{Conclusions}\label{sec:conclusions}

In this manuscript, we proposed marginal generalized raking (MGR) for estimating marginal estimands using parametric models. MGR generalizes the influence function-based theory of generalized raking for regression parameters to marginal estimands defined in parametric models. Our theoretical results show that MGR using the optimal calibration variable is asymptotically equivalent to the optimal augmented inverse probability of coarsening weighted estimator, generalizing a similar result for regression parameters. This is appealing for two reasons: generalized raking estimators are guaranteed to lie in the boundary of the parameter space, and generalized raking estimators are easily implemented in software \citep{surveypkg}. In simulations, the performance of our estimator matches what would be expected from the theory. 

In this parametric model setting, MGR was at least as efficient as conditional generalized raking (CGR) when using regression adjustment for confounders in all scenarios except when the propensity score model was misspecified but the outcome regression model was correctly specified. CGR uses marginalization after estimating regression parameters using generalized raking, and was pursued in previous research for estimating marginal parameters \citep{williamson2026assessing}. Here, we showed that directly raking on the efficient influence function for the marginal estimand provides a computationally simple and more robust approach compared to marginalizing the conditional estimator. In a data analysis, we saw that MGR can have substantial performance gains in limited-data settings, depending on the approach taken to adjust for confounding. 

While the parametric model setting sets an important baseline for optimal efficiency, the estimand of interest is not always defined with respect to a parametric model. In future research, we will consider an extension of this work to nonparametric models, where machine learning could be used to estimate nuisance functions (e.g., the outcome regression model, treatment propensity score, or missing-data model), providing a new computationally-efficient approach to multiply-robust estimation in missing-data settings. This increased robustness to model misspecification will likely come at a price of increased variance relative to the parametric model setting. 

\section*{Acknowledgments}
This manuscript was supported in part by the U.S. National Institutes of Health (NIH) grants R37AI131771 (BDW, RZ, TL, BES, PAS) and R01AG103351 (RZ, BDW, PAS). Data collection was supported by NIH grant P30AI110527 (BES). 

\section*{Disclosures}
BDW has worked on grants awarded to Kaiser Permanente Washington Health Research Institute (KPWHRI) by Eli Lilly.

\section*{Code availability}
We provide code to reproduce both simulations and data analyses on GitHub at \url{https://github.com/bdwilliamson/mgr_parametric_supplementary}.

\ifarXiv
\section*{Supplementary Material}
% start numbers with "S"
\renewcommand{\thefigure}{S\arabic{figure}}
\renewcommand{\theequation}{S\arabic{equation}}
\renewcommand{\thetable}{S\arabic{table}}
\renewcommand{\thetheorem}{S\arabic{theorm}}
\renewcommand{\thelemma}{S\arabic{lemma}}
\renewcommand{\thesection}{S\arabic{section}}
\renewcommand{\thealgorithm}{S\arabic{algorithm}}
\setcounter{figure}{0}
\setcounter{table}{0}
\setcounter{section}{0}
\setcounter{algorithm}{0}

\section{CGR (IPTW) estimator}\label{sec:cgr_iptw}

In the IPTW approach to CGR, we will use weights to account for the probability of receiving treatment, and the final outcome regression model will contain only $X$ (and an intercept). In other words, the ideal-data regression model is $E(Y \mid X) = h^{-1}\{(1,x)\alpha\}$, which we would estimate using IPTW weights in the absence of missing data. We use linear regression for continuous outcomes and logistic regression for binary outcomes. Linear regression could be used for binary outcomes as well but we chose not to pursue this to match with our CGR (RA) approach in the main manuscript. After obtaining $\pi_{n,IPCW}$, we estimate the treatment propensity score $g$ using logistic regression weighted by $\pi_{n,IPCW}^{-1}$, obtaining estimator $g_n$. 
In the IPTW approach, the optimal raking variable for the regression parameter $\alpha$ is $\eta_\alpha(v) = E\{\phi_\alpha^F(D) \mid V = v\}$, where $\phi_\alpha^F$ is the vector of EIFs for the regression parameters in the univariate regression model and $D = (L,A,X,Y)$ is the ideal data unit. Then, we use generalized raking to obtain an estimator $\alpha_{n,GR}$ of $\alpha$, modifying the procedure in the main manuscript to use an estimator $\eta_{\alpha,n}$ of $\eta_\alpha$ and initial weights $\pi_{n,IPCW,i}^{-1}$. Based on these estimators, with a slight abuse of notation, we define $Q_{n,1,CGR,IPTW}(x,l) = I(x = 1)g_n^{-1}(l)h^{-1}\{(1,x)\alpha_{n,GR}\}$ and $Q_{n,0,CGR,IPTW}(x,l) = I(x = 0)\{1 - g_n(l)\}^{-1}h^{-1}\{(1,x)\alpha_{n,GR}\}$. We provide the explicit steps for estimating $\mu_1$ using CGR-IPTW in Algorithm~\ref{alg:cgr_iptw}.
\begin{algorithm}
    \caption{Conditional generalized raking for the treatment-specific mean $\mu_1$ using inverse probability of treatment weighting}\label{alg:cgr_iptw}
    \begin{algorithmic}[1]
        \State construct estimator $\pi_{n,IPCW}$ of $\pi$ using generalized linear regression;
        \State construct estimator $\eta_{\alpha,n}$ of the optimal raking variable $\eta_{\alpha}$;
        \State calibrate the weights: set $\pi_{n,CGR-IPTW,i}^{-1} = b_i\pi_{n,IPCW,i}^{-1}$, where 
            \begin{align*}
                b^* =& \ \argmin_{b \in \mathbb{R}^n}\sum_{i=1}^n R_i d(b_i\pi_{n,IPCW,i}^{-1},\pi_{n,IPCW,i}^{-1}) \\
                & \text{ subject to }  \sum_{i=1}^n \eta_{\alpha,n}(V_i) = \sum_{i=1}^n R_i b_i \pi_{n,IPCW,i}^{-1} \eta_{\alpha,n}(V_i).
            \end{align*}
        \State define new weights $\pi_{n,CGR(I),i}^{-1} = \pi_{n,CGR-IPTW,i}^{-1}g_n(l_i)^{-1}$ to account for confounding using IPTW;
        \State obtain estimator $\alpha_{n,GR}$ using generalized linear regression of $Y$ on $X$ in the phase 2 subset weighted by $\pi_{n,GR,IPTW,i}^{-1}$
        \State construct estimator $Q_{n,1,CGR,IPTW}$ of $Q$ by setting $Q_{n,1,CGR,IPTW}(x,l) = I(x = 1)g_n^{-1}(l)h^{-1}\{(1,x)\alpha_{n,GR}\}$;
        \State compute estimator 
        \begin{align*}
            \mu_{1,n,CGR,IPTW} := \frac{1}{n}\sum_{i=1}^n \frac{R_i}{\pi_{n,CGR(I)}(V_i)} Q_{n,1,CGR}(1,L_i);
        \end{align*}
        \State compute estimator $\sigma^2_{1,n,CGR,IPTW}$ of the asymptotic variance $\sigma_{1,CGR,IPTW}^2$ of $n^{1/2}(\mu_{1,n,CGR,IPTW} - \mu_1)$ following Lemma~\ref{lemma:cgr_variance_iptw}: 
        \begin{align*}
            \sigma^2_{1,n,CGR,IPTW} =& \ \Omega_{2,n} \Sigma_{n,\alpha} \Omega_{2,n}^\top, \text{ where } \\
            \Omega_{2,n} =& \ \frac{1}{n}\sum_{i=1}^n \frac{R_i}{\pi_{n,CGR-IPTW}(V_i)}\frac{I(X_i = 1)}{g_n(L_i)}\nabla h^{-1}(1,1; \alpha_{n,GR}) \text{ and }\\
            \Sigma_{n,\alpha} =& \ \frac{1}{n}\sum_{i=1}^n \phi_\alpha(Y_i,1)\phi_\alpha(Y_i,1)^\top.
        \end{align*}
    \end{algorithmic}
\end{algorithm}

\section{Variance of the conditional generalized raking estimator}\label{sec:cgr_variance}
\subsection{Proof of Lemma 1}

Recall the assumptions and result from the main manuscript:
\begin{itemize}
    \item[(A1)] \textit{(positivity of the missing-data model)} $\pi(v) > \epsilon > 0$ for all $v$;
    \item[(A2)] \textit{(consistency of the regression model)} $Q_{n,CGR} \to_P Q$;
    \item[(A3)] \textit{(consistency of the missing-data model)} $\pi_{n,CGR} \to_P \pi$;
    \item[(A4)] \textit{(consistency of the EIF projection)}  $\eta_{\beta,n} \to_P \eta_\beta$.
\end{itemize}
If (A1), (A2), and either (A3) or (A4) hold, then 
\begin{align*}
    \sqrt{n}(\mu_{1,n,CGR} - \mu_1) \to_d& \ Z \sim N(0, \sigma^2_{1,CGR}), \text{ where } \\
    \sigma^2_{1,CGR} =& \ \Omega\Sigma_\beta \Omega^\top, \\
    \Omega =& \ E\left\{\frac{R}{\pi(V)}\nabla h^{-1}(1,1,L;\beta)\right\}, \text{ and } \\
    \Sigma_\beta =& \ E\{\phi_\beta(Y,X,L) \phi_\beta(Y,X,L)^\top\}.
\end{align*}

\begin{proof}
We first prove the result for regression adjustment. The generalized raking estimator $\beta_{n,GR}$ of the regression parameters $\beta$ is the solution to the estimating equation
\begin{align*}
    \frac{1}{n}\sum_{i=1}^n\frac{R_i}{\pi_{n,GR}(V_i)}S(Y_i,X_i,L_i;\tilde{\beta}) = 0,
\end{align*}
where $S(\cdot;\tilde{\beta})$ is the score function. Under (A1) and the assumption that either $\pi_n \to_P \pi$ (A3) or $\eta_{\beta,n} \to_P \eta_\beta$ (A4), $\sqrt{n}(\beta_{n,GR} - \beta) \to_d N(0, \Sigma_\beta)$, where $\Sigma_\beta = E\{\phi_\beta(Y,X,L,R,V) \phi_\beta(Y,X,L,R,V)^\top\}$. This is a consequence of the double-robustness of the generalized raking estimator of regression parameters \citep{lumley2011connections}. 

Now we want to describe the asymptotic variance of 
\begin{align*}
 \mu_{1,n,CGR} :=& \ \frac{1}{n}\sum_{i=1}^n \frac{R_i}{\pi_{n,CGR}(V_i)}Q_{n,CGR}(1,L_i),
\end{align*}
where $Q_{n,CGR}(x,l) = h^{-1}\{(x,l)\beta_{n,GR}\}$. We can also rewrite $\mu_{1} = E\{E(Y \mid X = 1, L)\}$ as a function of $\beta$ using the parametric model assumption on $Q$:
\begin{align*}
    \mu_{1}(\beta) =& \ E\left\{\frac{R}{\pi_0(V)}E_\beta(Y \mid X = 1, L)\right\} \\
    =& \ E\left[\frac{R}{\pi_0(V)}h^{-1}\{(1,L)\beta\}\right].
\end{align*}
Because $h^{-1}$ is a smooth function of $\beta$, and exchanging integration with differentiation, we can take a derivative of $\mu_1$ with respect to $\beta$, yielding 
\begin{align*}
    \nabla \mu_1(\beta') \mid_{\beta'=\beta} = E\left[\frac{R}{\pi_0(V)}\nabla h^{-1}\{(1,L)\beta'\}\mid_{\beta' = \beta}\right].
\end{align*}
Thus, in our study of the asymptotic distribution of $\mu_{1,n,CGR}$, we can apply the delta method: $\mu_{1,n,CGR}$ uses estimators of each of the true quantities in $\mu_1(\beta')$. By the weak law of large numbers, the empirical expectation (over the marginal distribution of $L$) converges to the true expectation, and because $Q_n \to_P Q$ (A2), we have that
\begin{align*}
    \sqrt{n}(\mu_{1,n,CGR} - \mu_1) \to_d N(0, \sigma^2_{1,CGR}),
\end{align*}
where $\sigma^2_{1,CGR} = \ \Omega \Sigma_\beta \Omega^\top$ and 
where $\Omega = \nabla \mu_1(\beta') \mid_{\beta'=\beta} = E\left[\frac{R}{\pi_0(V)}\nabla h^{-1}\{(1,L)\beta'\}\mid_{\beta' = \beta}\right]$,
precisely what we wished to show. 
\end{proof}

\subsection{Lemma 1 in the context of IPTW}
We first describe a modified version of (A2) and (A4) for use with IPTW: 
\begin{itemize}
    \item[(A2')] \textit{(consistency of the outcome regression model under IPTW)} Both \begin{itemize}
        \item[(A2'a)] $g_n \to_P g$, and
        \item[(A2'b)] $Q_{n,CGR,IPTW} \to_P Q$.
    \end{itemize}
    \item[(A4')]  \textit{(consistency of the EIF projection under IPTW)} $\eta_{\alpha,n} \to_P \eta_\alpha$. 
\end{itemize}
\begin{lemma}\label{lemma:cgr_variance_iptw}
    If (A1), (A2') and either (A3) or (A4') hold, then 
    \begin{align*}
        \sqrt{n}(\mu_{1,n,CGR,IPTW} - \mu_1) \to_d& \ Z \sim N(0, \sigma^2_{1,CGR,IPTW}), \text{ where } \\
        \sigma^2_{1,CGR,IPTW} =& \ \Omega_2\Sigma_\alpha \Omega_2^\top, \\
        \Omega_2 =& \ E\left\{\frac{R}{\pi(V)}\frac{I(X=1)}{g(L)}\nabla h^{-1}(1,1; \alpha)\right\}, \text{ and } \\
        \Sigma_\alpha =& \ E\{\phi_\alpha(Y,X,L) \phi_\alpha(Y,X,L)^\top\}.
    \end{align*}
\end{lemma}

\begin{proof}
    Under (A1) and either (A3) or (A4'), the generalized raking estimator $\alpha_{n,GR}$ of $\alpha$ has asymptotic behavior described by $\sqrt{n}(\alpha_{n,GR} - \alpha) \to_d N(0,\Sigma_\alpha)$, where $\Sigma_\alpha = E\{\phi_\alpha(Y,X,L,R,V)\allowbreak\phi_\alpha(Y,X,L,R,V)^\top\}$ and $\phi_\alpha^F$ incorporates inverse weighting by both $\pi$ and $g$. Following a similar argument as above, under (A2'), where now $Q_{n,CGR,IPTW}(x,l) = \frac{I(x=1)}{g_n(l)}h^{-1}\{(1,x)\alpha_{n,GR}\}$, the weak law of large numbers applies, and the delta method yields $\sigma_{1,CGR,IPTW}^2 = \Omega_2 \Sigma_\alpha \Omega_2^\top$, where 
\begin{align*}
    \Omega_2 = E\left[\frac{R}{\pi(V)}\frac{I(X=1)}{g(L)}\nabla h^{-1}(1,1; \alpha')\mid_{\alpha'=\alpha}\right].
\end{align*}
\end{proof}

When using IPTW, Assumption (A2') is equivalent to requiring that the propensity score model is consistent, i.e., $g_n \to_P g$; that the true propensity score is bounded away from 0 and 1; and that the inverse-weighted univariate regression model is consistent. Assumption (A4') requires that the estimated two-parameter optimal raking variable $\eta_{\alpha,n}$ is consistent. Assumptions (A2), (A4), (A2'), and (A4') are related: because $\eta_\beta$ and $\eta_\alpha$ are defined with respect to a parametric model, if Assumption (A2) or (A2') is not satisfied then it is more difficult to satisfy Assumption (A4) or (A4').

\subsection{Explicit variance for common link functions}

In this section, we provide the explicit form of the variance above for two common link functions: identity link (e.g., linear regression) and logit link (e.g., logistic regression). We provide the influence function here. Recall that for generalized linear models with information matrix $I(\beta)$ and score $S(\cdot; \beta)$, the efficient influence function is $I(\beta)^{-1}S(\cdot; \beta)$. For notational convenience, set $a = (x,l)$.

When $h$ is the identity link function, then $\nabla h^{-1}\{a\beta\} = a$ and $\phi_\beta(a) = I(\beta)^{-1}S(a,y; \beta) = (a^\top a)^{-1}a^\top(y - a\beta)$. 

When $h$ is the logit link function, then $h^{-1}(a\beta) = \{1+\exp(-a\beta)\}^{-1}$. In this case, $$\nabla h^{-1}(a\beta) = \frac{\exp(-a\beta)}{\{1+\exp(-a\beta)\}^2}a,$$ and the influence function is $$\phi_\beta(a) = I(\beta)^{-1}S(a,y;\beta) = \{a^\top w(\beta)a\}^{-1}a^\top \{y - h^{-1}(a\beta)\},$$
where $$w(\beta)= \frac{\exp(-a\beta)}{\{1+\exp(-a\beta)\}^2}.$$

\section{Marginal generalized raking}

\subsection{Proof of Theorem 1}

We first state and prove a lemma that the MGR estimator is asymptotically equivalent to the AIPCW-AIPTW estimator. 

\begin{lemma}\label{lem:aipw_mgr_equivalence}
    If (A1) holds, then $\mu_{1,n,MGR}$ is asymptotically equivalent to $\mu_{1,n,AIPCW-AIPTW}$.
\end{lemma}
\begin{proof}
    The AIPCW-AIPTW estimator is a solution to an estimating equation using the efficient influence function. Suppose that we have estimators $\pi_n$ of $\pi_0$, $g_n$ of $g_0$, $Q_n$ of $Q_0$, and $\eta_n$ of $\eta_0$. Recall that $\eta_0$ is a function of $\mu$ through the ideal-data influence function; write the uncentered version as $\eta_n^u(v) = \eta_n(v) + \mu$. Then we can write the estimating equation
    \begin{align*}
        0 =&\ \frac{1}{n}\sum_{i=1} \left(\frac{R_i}{\pi_n(V_i)}\left[\frac{I(X_i = 1)}{g_n(L_i)}\{Y_i - Q_n(1,L_i)\} + Q_n(1,L_i) - \mu\right] + \frac{\{\pi_n(V_i) - R_i\}}{\pi_n(V_i)}\{\eta_n^u(V_i) + \mu\}\right) \\
        =& \ \frac{1}{n}\sum_{i=1} \left(\frac{R_i}{\pi_n(V_i)}\left[\frac{I(X_i = 1)}{g_n(L_i)}\{Y_i - Q_n(1,L_i)\} + Q_n(1,L_i)\right] + \frac{\{\pi_n(V_i) - R_i\}}{\pi_n(V_i)}\eta_n^u(V_i)\right) \\
        &+ \frac{1}{n}\sum_{i=1}^n \left[-\frac{R_i}{\pi_n(V_i)}\mu + \left\{1 - \frac{R_i}{\pi_n(V_i)}\right\}\mu \right]
    \end{align*}
    This implies that
    \begin{align*}
        \mu =& \ \frac{1}{n}\sum_{i=1} \left(\frac{R_i}{\pi_n(V_i)}\left[\frac{I(X_i = 1)}{g_n(L_i)}\{Y_i - Q_n(1,L_i)\} + Q_n(1,L_i)\right] + \frac{\{\pi_n(V_i) - R_i\}}{\pi_n(V_i)}\eta_n^u(V_i)\right) \\
        =& \ \frac{1}{n}\sum_{i=1} \left(\frac{R_i}{\pi_n(V_i)}\left[\frac{I(X_i = 1)}{g_n(L_i)}\{Y_i - Q_n(1,L_i)\} + Q_n(1,L_i)\right] + \frac{\{\pi_n(V_i) - R_i\}}{\pi_n(V_i)}\eta_n(V_i)\right) \\
        & \qquad + o_P(n^{-1/2})\ ,
    \end{align*}
    which is the AIPCW-AIPTW estimator. 

    Now suppose that our estimator of $\pi_0$ is $\pi_{n,GR}$. Recall that $\pi_{n,GR}^{-1}(V_i) = b_i^*\pi_{n,IPCW}^{-1}(V_i)$ for each $i$ and that $\sum_{i=1}^n\eta_n(V_i) = \sum_{i=1}^n \frac{R_ib_i^*}{\pi_{n,IPCW}(V_i)}\eta_n(V_i)$. Plugging in $\pi_{n,GR}$ to the estimating equation above yields
    \begin{align*}
        \mu =& \ \frac{1}{n}\sum_{i=1} \left(\frac{R_i}{\pi_{n,GR}(V_i)}\left[\frac{I(X_i = 1)}{g_n(L_i)}\{Y_i - Q_n(1,L_i)\} + Q_n(1,L_i)\right] + \frac{\{\pi_{n,GR}(V_i) - R_i\}}{\pi_{n,GR}(V_i)}\eta_n(V_i)\right) \\
        &\qquad + o_P(n^{-1/2}) \\
        =& \ \frac{1}{n}\sum_{i=1} \left(\frac{R_i}{\pi_{n,GR}(V_i)}\left[\frac{I(X_i = 1)}{g_n(L_i)}\{Y_i - Q_n(1,L_i)\} + Q_n(1,L_i)\right] + \eta_n(V_i) - \frac{R_iw_i}{\pi_{n,IPCW}}\eta_n(V_i)\right) \\
        &\qquad + o_P(n^{-1/2}) \\
        =& \ \frac{1}{n}\sum_{i=1} \frac{R_i}{\pi_{n,GR}(V_i)}\left[\frac{I(X_i = 1)}{g_n(L_i)}\{Y_i - Q_n(1,L_i)\} + Q_n(1,L_i)\right] + o_P(n^{-1/2}),
    \end{align*}
    which is the MGR estimator.
\end{proof}
Next, we state and prove a lemma that the AIPCW-AIPTW estimator has a particular second-order remainder term.
\begin{lemma}\label{lem:aipcw-aiptw_remainder}
    Under (A1), we can write 
    \begin{align*}
        \mu_{1,n,AIPCW-AIPTW} - \mu_1 = \frac{1}{n}\sum_{i=1}^n\phi_{1,obs}(Y_i,X_i,L_i,R_i,V_i) + R(P_n,P_0) + o_P(n^{-1/2}),
    \end{align*}
    where $P_n$ includes estimators of all relevant quantities  and \begin{align*}
        R(P,P_0) =& \ E_0\left(\frac{\pi_P(V) - \pi_0(V)}{\pi_P(V)}\left[E_0\{\overline{\phi}_P^F(Y,X,L)\mid V\} - E_P\{\overline{\phi}_P^F(Y,X,L)\mid V\}\right] \right)\\
        &+ E_0\left[\left\{\frac{g_P(L) - g_0(L)}{g_P(L)}\right\}\{Q_P(1,X) - Q_0(1,X)\}\right].
    \end{align*}
\end{lemma}
\begin{proof}
    Consider the AIPCW-AIPTW estimator, which is the solution to estimating equations based on the EIF. Consider two distributions $P, P_0 \in \mathcal{M}$. Write $\mu(P) = E_P\{E_P(Y \mid A = 1, X)\}$; then $\mu_1 = \mu(P_0)$. The map $\mu: \mathcal{M} \to \mathbb{R}$ can be linearized at $P_0$ if we can write
    \begin{align*}
        \mu(P) - \mu(P_0) = (P - P_0)D(P) + R(P,P_0)
    \end{align*}
    for a gradient $D(P)$ of $\mu$ at $P$ and second-order remainder term $R$. In our case, we consider the efficient influence function $\phi_{1,obs}$ as the gradient \citep{bickel1993efficient}. Recall that the EIF evaluated at $P$ is 
    \begin{align*}
        \phi_{P}(r,v,y,x,l) =& \ \frac{r}{\pi_P(v)}\left[\frac{I(a = 1)}{g(l)}\{y - Q(1,x)\} + Q(1,x) - \mu(P)\right] \\
        &+ \frac{\pi_P(v) - r}{\pi_P(v)}E_P\{\phi_P^F(Y,X,L)\mid V = v\}, 
    \end{align*}
    where $\phi_P^F$ is the ideal-data EIF (i.e., the EIF where there is no missing data). Define $\overline{\phi}_P^F = \phi_P^F + \mu(P)$ as the uncentered EIF. Then we can write
    \begin{align*}
        \phi_{P}(r,v,y,x,l) =& \ \frac{r}{\pi_P(v)}\overline{\phi}_P^F(y,x,l) - \frac{r}{\pi_P(v)}\mu(P) \\
        &- \frac{r}{\pi_P(v)}E_P\{\overline{\phi}_P^F(Y,X,L)\mid V = v\} + \frac{r}{\pi_P(v)}\mu(P) \\
        &+ E_P\{\overline{\phi}_P^F(Y,X,L)\mid V = v\}- \mu(P) \\
        =& \ \frac{r}{\pi_P(v)}\left[\overline{\phi}_P^F(y,x,l) - E_P\{\overline{\phi}_P^F(Y,X,L)\mid V = v\}\right] \\
        &+ E_P\{\overline{\phi}_P^F(Y,X,L)\mid V = v\}- \mu(P).
    \end{align*}
    The linearization above implies that $R(P,P_0) = \mu(P) - \mu(P_0) + P_0\phi_P$; plugging in our expression above for the third term yields
    \begin{align*}
        P_0\phi_P =& \ E_0\left(\frac{R}{\pi_P(V)}\left[\overline{\phi}_P^F(Y,X,L) - E_P\{\overline{\phi}_P^F(Y,X,L)\mid V\}\right] \right.\\
        &\left. \ \ \ \ \ \ \  + E_P\{\overline{\phi}_P^F(Y,X,L)\mid V\}- \mu(P)\right).
    \end{align*}
    Thus,
    \begin{align*}
        R(P,P_0) =& \ \mu(P) - \mu(P_0) + P_0\phi_P \\
        =& \ \mu(P) - \mu(P_0) \\
        &+ E_0\left(\frac{R}{\pi_P(V)}\left[\overline{\phi}_P^F(Y,X,L) - E_P\{\overline{\phi}_P^F(Y,X,L)\mid V\}\right] \right.\\
        &\left. \ \ \ \ \ \ \ \ \ + E_P\{\overline{\phi}_P^F(Y,X,L)\mid V\}- \mu(P)\right) \\
        =& \ E_0\left(\frac{R}{\pi_P(V)}\left[\overline{\phi}_P^F(Y,X,L) - E_P\{\overline{\phi}_P^F(Y,X,L)\mid V\}\right] \right.\\
        &\left. \ \ \ \ \ \ \  + E_P\{\overline{\phi}_P^F(Y,X,L)\mid V\}\right) - \mu(P_0) \\
        =& \ E_0\left(\frac{R}{\pi_P(V)}\left[E_0\{\overline{\phi}_P^F(Y,X,L)\mid V\} - E_P\{\overline{\phi}_P^F(Y,X,L)\mid V\}\right] \right.\\
        &\left. \ \ \ \ \ \ \  + E_P\{\overline{\phi}_P^F(Y,X,L)\mid V\}\right) - \mu(P_0) \\
        =& \ E_0\left(\frac{\pi_0(V)}{\pi_P(V)}\left[E_0\{\overline{\phi}_P^F(Y,X,L)\mid V\} - E_P\{\overline{\phi}_P^F(Y,X,L)\mid V\}\right] \right.\\
        &\left. \ \ \ \ \ \ \  + E_P\{\overline{\phi}_P^F(Y,X,L)\mid V\}\right) - \mu(P_0) \\
        =& \ E_0\left(\frac{\pi_P(V) - \pi_0(V)}{\pi_P(V)}\left[E_0\{\overline{\phi}_P^F(Y,X,L)\mid V\} - E_P\{\overline{\phi}_P^F(Y,X,L)\mid V\}\right] \right.\\
        &\left. \ \ \ \ \ \ \  + E_0\{\overline{\phi}_P^F(Y,X,L)\mid V\}\right) - \mu(P_0).
    \end{align*}
    Now consider the ideal-data AIPTW estimator. We can pursue a similar linearization using the full-data EIF:
    \begin{align*}
        \mu(P) - \mu(P_0) =& \ R_2(P,P_0) - P_0 \phi_P^F;
    \end{align*}
    this implies that
    \begin{align*}
        R_2(P,P_0) =& \ \mu(P) - \mu(P_0) + E_0\{\phi_P^F(Y,X,L)\} \\
        =& \ \mu(P) - \mu(P_0) + E_0\{\overline{\phi}_P^F(Y,X,L)\} - \mu(P) \\
        =& \ E_0\{\overline{\phi}_P^F(Y,X,L)\} - \mu(P_0) \\
        =& \ E_0[E_0\{\overline{\phi}_P^F(Y,X,L)\mid V\}] - \mu(P_0).
    \end{align*}
    Thus, we can write
    \begin{align*}
        R(P,P_0) =& \ E_0\left(\frac{\pi_P(V) - \pi_0(V)}{\pi_P(V)}\left[E_0\{\overline{\phi}_P^F(Y,X,L)\mid V\} - E_P\{\overline{\phi}_P^F(Y,X,L)\mid V\}\right] \right)\\
        &+ R_2(P,P_0).
    \end{align*}
    Further, $R_2$ has been studied in the literature \citep{vanderlaan2011targeted}:
    \begin{align*}
        R_2(P,P_0) =& \ E_0\left[\left\{\frac{g_P(1,X) - g_0(1,X)}{g_P(1,X)}\right\}\{Q_P(1,X) - Q_0(1,X)\}\right].
    \end{align*}
    This is precisely what we wished to show.
\end{proof}

We now prove the theorem.
\begin{proof}[Proof of Theorem 1]
    By asymptotic equivalence of the MGR estimator and the AIPCW-AIPTW estimator from Lemma~\ref{lem:aipw_mgr_equivalence} and by the form of the remainder term in Lemma~\ref{lem:aipcw-aiptw_remainder}, we can write
    \begin{align*}
        \mu_{1,n,MGR} - \mu_1 = \frac{1}{n}\sum_{i=1}^n \phi_{1,obs}(Y_i,X_i,L_i,R_i,V_i) + R_n + o_P(n^{-1/2}),
    \end{align*}
    where
    \begin{align*}
        R_n =& \ E_0\left(\frac{\pi_n(V) - \pi_0(V)}{\pi_n(V)}\left[E_0\{\overline{\phi}_n^F(Y,X,L)\mid V\} - E_n\{\overline{\phi}_n^F(Y,X,L)\mid V\}\right] \right)\\
        &+ E_0\left[\left\{\frac{g_n(1,X) - g_0(1,X)}{g_n(1,X)}\right\}\{Q_n(1,X) - Q_0(1,X)\}\right].
    \end{align*}
    The form of $R_n$ yields the combination of double-robustness properties with respect to $g$ and $Q$ and $\pi$ and $\eta$. The second term yields the double-robustness enjoyed by a complete-data AIPTW or TMLE estimator \citep{vanderlaan2011targeted}. The first term yields a double-robustness property with respect to $\pi$ and $\eta$. This implies that if both [(B3) or (B4)] and [(B1) or (B2)] hold, i.e., both (i) either $\pi_n \to_P \pi_0$ or $\eta_n \to_P \eta_0$ and (ii) either $g_n \to_P g_0$ or $Q_n \to_P Q_0$ hold, then the MGR estimator is consistent. Recall that because all estimators are parametric, the convergence rate (if consistent) is $n^{-1/2}$. 
    Thus, under [(B3) or (B4)] and [(B1) or (B2)], $R_n = o_P(n^{-1/2})$, so the MGR estimator is asymptotically linear with influence function $\phi_{1,obs}$, which is the efficient influence function. Because this is the efficient influence function, the MGR estimator is efficient if all four nuisance functions are estimated consistently.
\end{proof}

\section{Additional numerical results}

In Table~\ref{tab:sim_model_coefs}, we present the data-generating model coefficients for each scenario. 
The results for correctly-specified models with only 20\% missing data (Scenario 1) are presented in Tables~\ref{tab:sim_combined_performance_tab_s1_ATE}--\ref{tab:sim_combined_performance_tab_s2_RR}.
The results for CGR (IPTW) in Scenario 2b are presented in Table~\ref{tab:sim_combined_performance_tab_s4_cgr_iptw}; results for Scenario 2a (continuous outcome) are in Table~\ref{tab:sim_combined_performance_tab_s3_ATE}.  Tables~\ref{tab:sim_combined_performance_tab_s5_ATE}--\ref{tab:sim_combined_performance_tab_s10_RR} show results in scenarios with misspecification. In Scenario 3 (Tables~\ref{tab:sim_combined_performance_tab_s5_ATE}--\ref{tab:sim_combined_performance_tab_s6_RR}), the fitted propensity score model is misspecified. In Scenario 4 (Tables~\ref{tab:sim_combined_performance_tab_s7_ATE}--\ref{tab:sim_combined_performance_tab_s8_RR}), the fitted outcome regression model is misspecified. In Scenario 6 (Tables~\ref{tab:sim_combined_performance_tab_s9_ATE}--\ref{tab:sim_combined_performance_tab_s10_RR}), the fitted missing-data model is misspecified.

% true coefficients
{
    % \fontsize{6}{8}\selectfont
    \footnotesize
    \begin{table}
\centering
\caption{Regression coefficients specifying the data-generating models for all simulation scenarios. \label{tab:sim_model_coefs}}
\centering
% \fontsize{11}{13}\selectfont
\fontsize{8}{10}\selectfont
\begin{tabular}[t]{llrlrrllrrrrl}
\toprule
\multicolumn{2}{c}{ } & \multicolumn{11}{c}{Coefficient} \\
\cmidrule(l{3pt}r{3pt}){3-13}
Scenario & Model & Intercept & X & Z1 & Z2 & W1 & W2 & U1 & U2 & C1 & C2 & Y\\
\midrule
 & Missing data & -1.335 & -0.1 & 0.05 & 0.0 & --- & --- & 0.0 & 0.0 & 0.00 & 0.00 & 0.05\\
\cmidrule{2-13}
 & Outcome regression & -0.250 & 1.5 & 1.00 & -1.0 & 1 & -1 & 0.0 & 0.0 & 0.00 & 0.00 & ---\\
\cmidrule{2-13}
\multirow{-3}{*}[2\dimexpr\aboverulesep+\belowrulesep+\cmidrulewidth]{\raggedright\arraybackslash 1a} & Propensity score & 0.100 & --- & 0.10 & 0.1 & 0.1 & 0.1 & 0.0 & 0.0 & 0.00 & 0.00 & ---\\
\cmidrule{1-13}
 & Missing data & -1.335 & -0.1 & 0.05 & 0.0 & --- & --- & 0.0 & 0.0 & 0.00 & 0.00 & 0.05\\
\cmidrule{2-13}
 & Outcome regression & -2.600 & 1.5 & 1.00 & -1.0 & 1 & -1 & 0.0 & 0.0 & 0.00 & 0.00 & ---\\
\cmidrule{2-13}
\multirow{-3}{*}[2\dimexpr\aboverulesep+\belowrulesep+\cmidrulewidth]{\raggedright\arraybackslash 1b} & Propensity score & 0.100 & --- & 0.10 & 0.1 & 0.1 & 0.1 & 0.0 & 0.0 & 0.00 & 0.00 & ---\\
\cmidrule{1-13}
 & Missing data & 0.300 & -0.1 & 0.05 & 0.0 & --- & --- & 0.0 & 0.0 & -0.10 & -0.10 & 0.025\\
\cmidrule{2-13}
 & Outcome regression & -3.000 & 1.5 & 1.00 & -1.0 & 1 & -1 & 0.0 & 0.0 & 0.75 & -0.75 & ---\\
\cmidrule{2-13}
\multirow{-3}{*}[2\dimexpr\aboverulesep+\belowrulesep+\cmidrulewidth]{\raggedright\arraybackslash 2a} & Propensity score & 0.100 & --- & 0.10 & 0.1 & 0.1 & 0.1 & 0.0 & 0.0 & -0.40 & -0.20 & ---\\
\cmidrule{1-13}
 & Missing data & 0.300 & -0.1 & 0.05 & 0.0 & --- & --- & 0.0 & 0.0 & -0.10 & -0.10 & 0.05\\
\cmidrule{2-13}
 & Outcome regression & -3.000 & 1.5 & 1.00 & -1.0 & 1 & -1 & 0.0 & 0.0 & 0.75 & -0.75 & ---\\
\cmidrule{2-13}
\multirow{-3}{*}[2\dimexpr\aboverulesep+\belowrulesep+\cmidrulewidth]{\raggedright\arraybackslash 2b} & Propensity score & 0.100 & --- & 0.10 & 0.1 & 0.1 & 0.1 & 0.0 & 0.0 & -0.40 & -0.20 & ---\\
\cmidrule{1-13}
 & Missing data & 0.300 & -0.1 & 0.05 & 0.0 & --- & --- & 0.0 & 0.0 & -0.10 & -0.10 & 0.025\\
\cmidrule{2-13}
 & Outcome regression & -8.000 & 1.5 & 1.00 & -1.0 & 1 & -1 & 1.5 & -1.0 & 0.75 & -0.75 & ---\\
\cmidrule{2-13}
\multirow{-3}{*}[2\dimexpr\aboverulesep+\belowrulesep+\cmidrulewidth]{\raggedright\arraybackslash 3a} & Propensity score & 0.100 & --- & 0.10 & 0.1 & 0.1 & 0.1 & -1.0 & -1.0 & -0.40 & -0.20 & ---\\
\cmidrule{1-13}
 & Missing data & 0.300 & -0.1 & 0.05 & 0.0 & --- & --- & 0.0 & 0.0 & -0.10 & -0.10 & 0.05\\
\cmidrule{2-13}
 & Outcome regression & -8.000 & 1.5 & 1.00 & -1.0 & 1 & -1 & 3.0 & -5.0 & 0.75 & -0.75 & ---\\
\cmidrule{2-13}
\multirow{-3}{*}[2\dimexpr\aboverulesep+\belowrulesep+\cmidrulewidth]{\raggedright\arraybackslash 3b} & Propensity score & 0.100 & --- & 0.10 & 0.1 & 0.1 & 0.1 & -1.0 & -1.0 & -0.40 & -0.20 & ---\\
\cmidrule{1-13}
 & Missing data & 0.300 & -0.1 & 0.05 & 0.0 & --- & --- & 0.0 & 0.0 & -0.10 & -0.10 & 0.025\\
\cmidrule{2-13}
 & Outcome regression & -8.000 & 1.5 & 1.00 & -1.0 & 1 & -1 & 1.5 & -1.0 & 0.75 & -0.75 & ---\\
\cmidrule{2-13}
\multirow{-3}{*}[2\dimexpr\aboverulesep+\belowrulesep+\cmidrulewidth]{\raggedright\arraybackslash 4a} & Propensity score & 0.100 & --- & 0.10 & 0.1 & 0.1 & 0.1 & -1.0 & -1.0 & -0.40 & -0.20 & ---\\
\cmidrule{1-13}
 & Missing data & 0.300 & -0.1 & 0.05 & 0.0 & --- & --- & 0.0 & 0.0 & -0.10 & -0.10 & 0.05\\
\cmidrule{2-13}
 & Outcome regression & -8.000 & 1.5 & 1.00 & -1.0 & 1 & -1 & 3.0 & -5.0 & 0.75 & -0.75 & ---\\
\cmidrule{2-13}
\multirow{-3}{*}[2\dimexpr\aboverulesep+\belowrulesep+\cmidrulewidth]{\raggedright\arraybackslash 4b} & Propensity score & 0.100 & --- & 0.10 & 0.1 & 0.1 & 0.1 & -1.0 & -1.0 & -0.40 & -0.20 & ---\\
\cmidrule{1-13}
5 & Propensity score & 0.100 & --- & 0.10 & 0.1 & 0.1 & 0.1 & 0.0 & 0.0 & -0.40 & -0.20 & ---\\
\cmidrule{1-13}
 & Missing data & 0.300 & -0.1 & 0.05 & 0.0 & --- & --- & -0.5 & 0.5 & -0.10 & -0.10 & 0.025\\
\cmidrule{2-13}
\multirow{-2}{*}[1\dimexpr\aboverulesep+\belowrulesep+\cmidrulewidth]{\raggedright\arraybackslash 5a} & Outcome regression & -3.000 & 1.5 & 1.00 & -1.0 & 1 & -1 & 0.0 & 0.0 & 0.75 & -0.75 & ---\\
\cmidrule{1-13}
 & Missing data & 0.300 & -0.1 & 0.05 & 0.0 & --- & --- & -0.5 & 0.5 & -0.10 & -0.10 & 0.05\\
\cmidrule{2-13}
\multirow{-2}{*}[1\dimexpr\aboverulesep+\belowrulesep+\cmidrulewidth]{\raggedright\arraybackslash 5b} & Outcome regression & -3.000 & 1.5 & 1.00 & -1.0 & 1 & -1 & 0.0 & 0.0 & 0.75 & -0.75 & ---\\
\cmidrule{1-13}
6a & Propensity score & 0.100 & --- & 0.10 & 0.1 & 0.1 & 0.1 & -1.0 & -1.0 & -0.40 & -0.20 & ---\\
\cmidrule{1-13}
6b & Propensity score & 0.100 & --- & 0.10 & 0.1 & 0.1 & 0.1 & -1.0 & -1.0 & -0.40 & -0.20 & ---\\
\bottomrule
\end{tabular}
\end{table}

}

% correctly-specified scenarios
\begin{table}[!h]
\centering
\caption{Performance of CGR (IPTW) estimators of the ATE and RR in Scenario 2b, a binary outcome setting with approximately 50\% missing data and correctly specified nuisance models.\label{tab:sim_combined_performance_tab_s4_cgr_iptw}}
\centering
\fontsize{9}{11}\selectfont
\begin{threeparttable}
\begin{tabular}[t]{>{\raggedright\arraybackslash}p{3.5em}lrrrrrrrrr}
\toprule
Estimand & n & Med. bias & ASE & ESE & MAD & rRMSE & RMSE & Rel. Eff. & Cover & ESE Cover\\
\midrule
ATE & 1000 & -0.0014 & 0.021 & 0.021 & 0.022 & 0.022 & 0.021 & 0.860 & 0.949 & 0.948\\
\cmidrule{1-11}
ATE & 2000 & 0.0003 & 0.015 & 0.015 & 0.015 & 0.015 & 0.015 & 0.909 & 0.944 & 0.947\\
\cmidrule{1-11}
ATE & 4000 & -0.0003 & 0.010 & 0.010 & 0.010 & 0.010 & 0.010 & 0.927 & 0.945 & 0.949\\
\cmidrule{1-11}
ATE & 8000 & -0.0004 & 0.007 & 0.007 & 0.007 & 0.007 & 0.007 & 0.912 & 0.950 & 0.950\\
\cmidrule{1-11}
RR & 1000 & -0.0098 & 0.213 & 0.214 & 0.214 & 0.214 & 0.220 & 0.921 & 0.950 & 0.952\\
\cmidrule{1-11}
RR & 2000 & 0.0162 & 0.149 & 0.150 & 0.147 & 0.148 & 0.157 & 0.954 & 0.945 & 0.947\\
\cmidrule{1-11}
RR & 4000 & -0.0094 & 0.105 & 0.106 & 0.104 & 0.105 & 0.106 & 0.926 & 0.942 & 0.945\\
\cmidrule{1-11}
RR & 8000 & -0.0092 & 0.074 & 0.074 & 0.072 & 0.073 & 0.074 & 0.932 & 0.946 & 0.948\\
\bottomrule
\end{tabular}
\begin{tablenotes}
\item Abbreviations: Med: Median; ASE: median asymptotic standard error; ESE: empirical standard error; MAD: mean absolute deviation; rRMSE: robust root mean squared error (RMSE) (square root of med bias squared + MAD squared); RMSE: RMSE (mean bias squared + ESE squared); Rel. Eff.: relative efficiency (rRMSE relative to that of MGR); Cover: coverage based on the ASE; ESE Cover: coverage based on the ESE; CGR: conditional generalized raking (GR); IPTW: inverse probability of treatment weighting. Results are based on 2500 Monte-Carlo replications. 
\end{tablenotes}
\end{threeparttable}
\end{table}

\begin{table}[!h]
\centering
\caption{Performance of estimators of the ATE in Scenario 2a, a continuous outcome setting with approximately 50\% missing data and correctly specified nuisance models.\label{tab:sim_combined_performance_tab_s3_ATE}}
\centering
\fontsize{9}{11}\selectfont
\begin{threeparttable}
\begin{tabular}[t]{>{\raggedright\arraybackslash}p{3.5em}lrrrrrrrrr}
\toprule
Est. & n & Med. bias & ASE & ESE & MAD & rRMSE & RMSE & Rel. Eff. & Cover & ESE Cover\\
\midrule
AIPCW-AIPTW &  & -0.0028 & 0.212 & 0.203 & 0.201 & 0.201 & 0.203 & 0.931 & 0.962 & 0.951\\
\cmidrule{1-1}
\cmidrule{3-11}
CGR (IPTW) &  & -0.0153 & 0.342 & 0.353 & 0.340 & 0.340 & 0.353 & 0.551 & 0.960 & 0.946\\
\cmidrule{1-1}
\cmidrule{3-11}
CGR (RA) &  & 0.0085 & 0.171 & 0.180 & 0.180 & 0.180 & 0.180 & 1.038 & 0.940 & 0.948\\
\cmidrule{1-1}
\cmidrule{3-11}
MGR & \multirow{-4}{*}[3\dimexpr\aboverulesep+\belowrulesep+\cmidrulewidth]{\raggedright\arraybackslash 250} & 0.0052 & 0.213 & 0.187 & 0.187 & 0.187 & 0.187 & 1.000 & 0.972 & 0.949\\
\cmidrule{1-11}
AIPCW-AIPTW &  & 0.0006 & 0.136 & 0.132 & 0.136 & 0.136 & 0.132 & 0.948 & 0.960 & 0.952\\
\cmidrule{1-1}
\cmidrule{3-11}
CGR (IPTW) &  & -0.0010 & 0.237 & 0.234 & 0.225 & 0.225 & 0.234 & 0.571 & 0.965 & 0.951\\
\cmidrule{1-1}
\cmidrule{3-11}
CGR (RA) &  & 0.0023 & 0.121 & 0.123 & 0.125 & 0.125 & 0.123 & 1.027 & 0.948 & 0.952\\
\cmidrule{1-1}
\cmidrule{3-11}
MGR & \multirow{-4}{*}[3\dimexpr\aboverulesep+\belowrulesep+\cmidrulewidth]{\raggedright\arraybackslash 500} & 0.0006 & 0.136 & 0.127 & 0.129 & 0.129 & 0.127 & 1.000 & 0.968 & 0.949\\
\cmidrule{1-11}
AIPCW-AIPTW &  & -0.0022 & 0.091 & 0.089 & 0.090 & 0.090 & 0.089 & 0.955 & 0.958 & 0.953\\
\cmidrule{1-1}
\cmidrule{3-11}
CGR (IPTW) &  & -0.0061 & 0.166 & 0.161 & 0.153 & 0.153 & 0.161 & 0.558 & 0.960 & 0.950\\
\cmidrule{1-1}
\cmidrule{3-11}
CGR (RA) &  & -0.0030 & 0.086 & 0.086 & 0.084 & 0.084 & 0.086 & 1.023 & 0.953 & 0.950\\
\cmidrule{1-1}
\cmidrule{3-11}
MGR & \multirow{-4}{*}[3\dimexpr\aboverulesep+\belowrulesep+\cmidrulewidth]{\raggedright\arraybackslash 1000} & -0.0011 & 0.092 & 0.087 & 0.085 & 0.085 & 0.087 & 1.000 & 0.964 & 0.951\\
\cmidrule{1-11}
AIPCW-AIPTW &  & -0.0017 & 0.063 & 0.064 & 0.064 & 0.064 & 0.064 & 0.963 & 0.949 & 0.954\\
\cmidrule{1-1}
\cmidrule{3-11}
CGR (IPTW) &  & 0.0031 & 0.117 & 0.116 & 0.116 & 0.116 & 0.116 & 0.531 & 0.953 & 0.951\\
\cmidrule{1-1}
\cmidrule{3-11}
CGR (RA) &  & -0.0009 & 0.061 & 0.062 & 0.062 & 0.062 & 0.062 & 1.001 & 0.942 & 0.949\\
\cmidrule{1-1}
\cmidrule{3-11}
MGR & \multirow{-4}{*}[3\dimexpr\aboverulesep+\belowrulesep+\cmidrulewidth]{\raggedright\arraybackslash 2000} & -0.0009 & 0.063 & 0.063 & 0.062 & 0.062 & 0.063 & 1.000 & 0.952 & 0.951\\
\bottomrule
\end{tabular}
\begin{tablenotes}
\item Abbreviations: ATE: average treatment effect; RR: relative risk; Est: Estimator; Med: Median; ASE: median asymptotic standard error; ESE: empirical standard error; MAD: mean absolute deviation; rRMSE: robust root mean squared error (RMSE) (square root of med bias squared + MAD squared); RMSE: RMSE (mean bias squared + ESE squared); Rel. Eff.: relative efficiency (rRMSE relative to that of MGR); Cover: coverage based on the ASE; ESE Cover: coverage based on the ESE; CGR: conditional generalized raking (GR); IPTW: inverse probability of treatment weighting; RA: regression adjustment; MGR: marginal GR; AIPCW: augmented inverse probability of coarsening weighting; AIPTW: augmented IPTW. Results are based on 2500 Monte-Carlo replications. 
\end{tablenotes}
\end{threeparttable}
\end{table}

\begin{table}[!h]
\centering
\caption{Performance of estimators of the ATE in Scenario 1a, a continuous outcome setting with approximately 20\% missing data and correctly specified nuisance models.\label{tab:sim_combined_performance_tab_s1_ATE}}
\centering
\fontsize{9}{11}\selectfont
\begin{threeparttable}
\begin{tabular}[t]{>{\raggedright\arraybackslash}p{3.5em}lrrrrrrrrr}
\toprule
Est. & n & Med. bias & ASE & ESE & MAD & rRMSE & RMSE & Rel. Eff. & Cover & ESE Cover\\
\midrule
AIPCW-AIPTW &  & -0.0058 & 0.149 & 0.143 & 0.148 & 0.149 & 0.143 & 0.945 & 0.960 & 0.948\\
\cmidrule{1-1}
\cmidrule{3-11}
CGR (IPTW) &  & -0.0021 & 0.290 & 0.186 & 0.190 & 0.190 & 0.186 & 0.738 & 0.998 & 0.954\\
\cmidrule{1-1}
\cmidrule{3-11}
CGR (RA) &  & -0.0048 & 0.138 & 0.141 & 0.144 & 0.144 & 0.141 & 0.972 & 0.946 & 0.950\\
\cmidrule{1-1}
\cmidrule{3-11}
MGR & \multirow{-4}{*}[3\dimexpr\aboverulesep+\belowrulesep+\cmidrulewidth]{\raggedright\arraybackslash 250} & -0.0030 & 0.150 & 0.140 & 0.140 & 0.140 & 0.140 & 1.000 & 0.963 & 0.951\\
\cmidrule{1-11}
AIPCW-AIPTW &  & 0.0013 & 0.102 & 0.101 & 0.100 & 0.100 & 0.101 & 0.984 & 0.950 & 0.948\\
\cmidrule{1-1}
\cmidrule{3-11}
CGR (IPTW) &  & 0.0022 & 0.203 & 0.132 & 0.126 & 0.126 & 0.132 & 0.783 & 0.996 & 0.946\\
\cmidrule{1-1}
\cmidrule{3-11}
CGR (RA) &  & 0.0033 & 0.098 & 0.100 & 0.100 & 0.100 & 0.100 & 0.989 & 0.944 & 0.948\\
\cmidrule{1-1}
\cmidrule{3-11}
MGR & \multirow{-4}{*}[3\dimexpr\aboverulesep+\belowrulesep+\cmidrulewidth]{\raggedright\arraybackslash 500} & 0.0023 & 0.102 & 0.100 & 0.099 & 0.099 & 0.100 & 1.000 & 0.954 & 0.949\\
\cmidrule{1-11}
AIPCW-AIPTW &  & -0.0011 & 0.071 & 0.069 & 0.070 & 0.070 & 0.069 & 1.000 & 0.953 & 0.948\\
\cmidrule{1-1}
\cmidrule{3-11}
CGR (IPTW) &  & -0.0028 & 0.143 & 0.090 & 0.090 & 0.090 & 0.090 & 0.776 & 0.997 & 0.947\\
\cmidrule{1-1}
\cmidrule{3-11}
CGR (RA) &  & -0.0007 & 0.069 & 0.069 & 0.070 & 0.070 & 0.069 & 1.000 & 0.948 & 0.947\\
\cmidrule{1-1}
\cmidrule{3-11}
MGR & \multirow{-4}{*}[3\dimexpr\aboverulesep+\belowrulesep+\cmidrulewidth]{\raggedright\arraybackslash 1000} & -0.0011 & 0.071 & 0.069 & 0.070 & 0.070 & 0.069 & 1.000 & 0.952 & 0.949\\
\cmidrule{1-11}
AIPCW-AIPTW &  & 0.0008 & 0.050 & 0.048 & 0.047 & 0.047 & 0.048 & 1.015 & 0.958 & 0.950\\
\cmidrule{1-1}
\cmidrule{3-11}
CGR (IPTW) &  & -0.0015 & 0.101 & 0.064 & 0.065 & 0.065 & 0.064 & 0.743 & 0.998 & 0.948\\
\cmidrule{1-1}
\cmidrule{3-11}
CGR (RA) &  & 0.0003 & 0.049 & 0.048 & 0.048 & 0.048 & 0.048 & 1.009 & 0.955 & 0.950\\
\cmidrule{1-1}
\cmidrule{3-11}
MGR & \multirow{-4}{*}[3\dimexpr\aboverulesep+\belowrulesep+\cmidrulewidth]{\raggedright\arraybackslash 2000} & 0.0012 & 0.050 & 0.048 & 0.048 & 0.048 & 0.048 & 1.000 & 0.958 & 0.950\\
\cmidrule{1-11}
AIPCW-AIPTW &  & -0.0011 & 0.031 & 0.031 & 0.031 & 0.031 & 0.031 & 0.978 & 0.951 & 0.949\\
\cmidrule{1-1}
\cmidrule{3-11}
CGR (IPTW) &  & -0.0003 & 0.064 & 0.041 & 0.041 & 0.041 & 0.041 & 0.760 & 0.999 & 0.948\\
\cmidrule{1-1}
\cmidrule{3-11}
CGR (RA) &  & -0.0008 & 0.031 & 0.031 & 0.030 & 0.030 & 0.031 & 1.013 & 0.949 & 0.949\\
\cmidrule{1-1}
\cmidrule{3-11}
MGR & \multirow{-4}{*}[3\dimexpr\aboverulesep+\belowrulesep+\cmidrulewidth]{\raggedright\arraybackslash 5000} & -0.0003 & 0.031 & 0.031 & 0.031 & 0.031 & 0.031 & 1.000 & 0.950 & 0.946\\
\bottomrule
\end{tabular}
\begin{tablenotes}
\item Abbreviations: ATE: average treatment effect; Est: Estimator; Med: Median; ASE: median asymptotic standard error; ESE: empirical standard error; MAD: mean absolute deviation; rRMSE: robust root mean squared error (RMSE) (square root of med bias squared + MAD squared); RMSE: RMSE (mean bias squared + ESE squared); Rel. Eff.: relative efficiency (rRMSE relative to that of MGR); Cover: coverage based on the ASE; ESE Cover: coverage based on the ESE; CGR: conditional generalized raking (GR); IPTW: inverse probability of treatment weighting; RA: regression adjustment; MGR: marginal GR; AIPCW: augmented inverse probability of coarsening weighting; AIPTW: augmented IPTW. Results are based on 2500 Monte-Carlo replications. 
\end{tablenotes}
\end{threeparttable}
\end{table}

\begin{table}[!h]
\centering
\caption{Performance of estimators of the ATE in Scenario 1b, a binary outcome setting with approximately 20\% missing data and correctly specified nuisance models.\label{tab:sim_combined_performance_tab_s2_ATE}}
\centering
\fontsize{9}{11}\selectfont
\begin{threeparttable}
\begin{tabular}[t]{>{\raggedright\arraybackslash}p{3.5em}lrrrrrrrrr}
\toprule
Est. & n & Med. bias & ASE & ESE & MAD & rRMSE & RMSE & Rel. Eff. & Cover & ESE Cover\\
\midrule
AIPCW-AIPTW &  & $<$ 0.0001 & 0.023 & 0.023 & 0.023 & 0.023 & 0.023 & 1.003 & 0.950 & 0.950\\
\cmidrule{1-1}
\cmidrule{3-11}
CGR (IPTW) &  & -2 $\times 10^{-04}$ & 0.027 & 0.024 & 0.025 & 0.025 & 0.024 & 0.926 & 0.974 & 0.952\\
\cmidrule{1-1}
\cmidrule{3-11}
CGR (RA) &  & -4 $\times 10^{-04}$ & 0.022 & 0.023 & 0.023 & 0.023 & 0.023 & 0.996 & 0.940 & 0.947\\
\cmidrule{1-1}
\cmidrule{3-11}
MGR & \multirow{-4}{*}[3\dimexpr\aboverulesep+\belowrulesep+\cmidrulewidth]{\raggedright\arraybackslash 1000} & -3 $\times 10^{-04}$ & 0.023 & 0.023 & 0.023 & 0.023 & 0.023 & 1.000 & 0.947 & 0.948\\
\cmidrule{1-11}
AIPCW-AIPTW &  & 4 $\times 10^{-04}$ & 0.016 & 0.016 & 0.016 & 0.016 & 0.016 & 1.005 & 0.954 & 0.954\\
\cmidrule{1-1}
\cmidrule{3-11}
CGR (IPTW) &  & -4 $\times 10^{-04}$ & 0.019 & 0.017 & 0.017 & 0.017 & 0.017 & 0.952 & 0.978 & 0.953\\
\cmidrule{1-1}
\cmidrule{3-11}
CGR (RA) &  & 2 $\times 10^{-04}$ & 0.016 & 0.016 & 0.016 & 0.016 & 0.016 & 0.994 & 0.950 & 0.950\\
\cmidrule{1-1}
\cmidrule{3-11}
MGR & \multirow{-4}{*}[3\dimexpr\aboverulesep+\belowrulesep+\cmidrulewidth]{\raggedright\arraybackslash 2000} & 2 $\times 10^{-04}$ & 0.016 & 0.016 & 0.016 & 0.016 & 0.016 & 1.000 & 0.954 & 0.954\\
\cmidrule{1-11}
AIPCW-AIPTW &  & -3 $\times 10^{-04}$ & 0.011 & 0.011 & 0.011 & 0.011 & 0.011 & 0.985 & 0.954 & 0.947\\
\cmidrule{1-1}
\cmidrule{3-11}
CGR (IPTW) &  & -6 $\times 10^{-04}$ & 0.013 & 0.011 & 0.011 & 0.011 & 0.011 & 0.924 & 0.980 & 0.949\\
\cmidrule{1-1}
\cmidrule{3-11}
CGR (RA) &  & -7 $\times 10^{-04}$ & 0.011 & 0.011 & 0.010 & 0.010 & 0.011 & 0.989 & 0.955 & 0.950\\
\cmidrule{1-1}
\cmidrule{3-11}
MGR & \multirow{-4}{*}[3\dimexpr\aboverulesep+\belowrulesep+\cmidrulewidth]{\raggedright\arraybackslash 4000} & -7 $\times 10^{-04}$ & 0.011 & 0.011 & 0.010 & 0.010 & 0.011 & 1.000 & 0.954 & 0.946\\
\cmidrule{1-11}
AIPCW-AIPTW &  & 2 $\times 10^{-04}$ & 0.008 & 0.008 & 0.008 & 0.008 & 0.008 & 0.973 & 0.941 & 0.954\\
\cmidrule{1-1}
\cmidrule{3-11}
CGR (IPTW) &  & $<$ 0.0001 & 0.010 & 0.009 & 0.009 & 0.009 & 0.009 & 0.922 & 0.968 & 0.946\\
\cmidrule{1-1}
\cmidrule{3-11}
CGR (RA) &  & -1 $\times 10^{-04}$ & 0.008 & 0.008 & 0.008 & 0.008 & 0.008 & 0.987 & 0.943 & 0.952\\
\cmidrule{1-1}
\cmidrule{3-11}
MGR & \multirow{-4}{*}[3\dimexpr\aboverulesep+\belowrulesep+\cmidrulewidth]{\raggedright\arraybackslash 8000} & -1 $\times 10^{-04}$ & 0.008 & 0.008 & 0.008 & 0.008 & 0.008 & 1.000 & 0.938 & 0.950\\
\bottomrule
\end{tabular}
\begin{tablenotes}
\item Abbreviations: ATE: average treatment effect; RR: relative risk; Est: Estimator; Med: Median; ASE: median asymptotic standard error; ESE: empirical standard error; MAD: mean absolute deviation; rRMSE: robust root mean squared error (RMSE) (square root of med bias squared + MAD squared); RMSE: RMSE (mean bias squared + ESE squared); Rel. Eff.: relative efficiency (rRMSE relative to that of MGR); Cover: coverage based on the ASE; ESE Cover: coverage based on the ESE; CGR: conditional generalized raking (GR); IPTW: inverse probability of treatment weighting; RA: regression adjustment; MGR: marginal GR; AIPCW: augmented inverse probability of coarsening weighting; AIPTW: augmented IPTW. Results are based on 2500 Monte-Carlo replications. 
\end{tablenotes}
\end{threeparttable}
\end{table}

\begin{table}[!h]
\centering
\caption{Performance of estimators of the RR in Scenario 1b, a binary outcome setting with approximately 20\% missing data and correctly specified nuisance models.\label{tab:sim_combined_performance_tab_s2_RR}}
\centering
\fontsize{9}{11}\selectfont
\begin{threeparttable}
\begin{tabular}[t]{>{\raggedright\arraybackslash}p{3.5em}lrrrrrrrrr}
\toprule
Est. & n & Med. bias & ASE & ESE & MAD & rRMSE & RMSE & Rel. Eff. & Cover & ESE Cover\\
\midrule
AIPCW-AIPTW &  & -0.0035 & 0.104 & 0.104 & 0.102 & 0.102 & 0.105 & 1.001 & 0.948 & 0.951\\
\cmidrule{1-1}
\cmidrule{3-11}
CGR (IPTW) &  & -0.0014 & 0.122 & 0.110 & 0.112 & 0.112 & 0.112 & 0.911 & 0.974 & 0.952\\
\cmidrule{1-1}
\cmidrule{3-11}
CGR (RA) &  & 0.0016 & 0.102 & 0.104 & 0.102 & 0.102 & 0.106 & 1.000 & 0.948 & 0.948\\
\cmidrule{1-1}
\cmidrule{3-11}
MGR & \multirow{-4}{*}[3\dimexpr\aboverulesep+\belowrulesep+\cmidrulewidth]{\raggedright\arraybackslash 1000} & 0.0030 & 0.105 & 0.104 & 0.102 & 0.102 & 0.106 & 1.000 & 0.952 & 0.944\\
\cmidrule{1-11}
AIPCW-AIPTW &  & -0.0018 & 0.073 & 0.073 & 0.070 & 0.070 & 0.073 & 1.012 & 0.948 & 0.949\\
\cmidrule{1-1}
\cmidrule{3-11}
CGR (IPTW) &  & 0.0002 & 0.086 & 0.077 & 0.075 & 0.075 & 0.078 & 0.954 & 0.973 & 0.950\\
\cmidrule{1-1}
\cmidrule{3-11}
CGR (RA) &  & 0.0034 & 0.072 & 0.073 & 0.070 & 0.071 & 0.073 & 1.009 & 0.948 & 0.950\\
\cmidrule{1-1}
\cmidrule{3-11}
MGR & \multirow{-4}{*}[3\dimexpr\aboverulesep+\belowrulesep+\cmidrulewidth]{\raggedright\arraybackslash 2000} & 0.0039 & 0.073 & 0.073 & 0.071 & 0.071 & 0.074 & 1.000 & 0.951 & 0.950\\
\cmidrule{1-11}
AIPCW-AIPTW &  & -0.0050 & 0.051 & 0.050 & 0.048 & 0.048 & 0.050 & 0.987 & 0.955 & 0.948\\
\cmidrule{1-1}
\cmidrule{3-11}
CGR (IPTW) &  & 0.0015 & 0.061 & 0.053 & 0.053 & 0.053 & 0.053 & 0.896 & 0.972 & 0.945\\
\cmidrule{1-1}
\cmidrule{3-11}
CGR (RA) &  & -0.0009 & 0.051 & 0.050 & 0.047 & 0.047 & 0.050 & 0.998 & 0.953 & 0.948\\
\cmidrule{1-1}
\cmidrule{3-11}
MGR & \multirow{-4}{*}[3\dimexpr\aboverulesep+\belowrulesep+\cmidrulewidth]{\raggedright\arraybackslash 4000} & -0.0015 & 0.051 & 0.050 & 0.047 & 0.047 & 0.050 & 1.000 & 0.956 & 0.947\\
\cmidrule{1-11}
AIPCW-AIPTW &  & -0.0009 & 0.036 & 0.037 & 0.037 & 0.037 & 0.037 & 0.983 & 0.946 & 0.954\\
\cmidrule{1-1}
\cmidrule{3-11}
CGR (IPTW) &  & 0.0028 & 0.043 & 0.040 & 0.039 & 0.039 & 0.040 & 0.943 & 0.964 & 0.947\\
\cmidrule{1-1}
\cmidrule{3-11}
CGR (RA) &  & 0.0041 & 0.036 & 0.037 & 0.037 & 0.037 & 0.038 & 0.990 & 0.944 & 0.954\\
\cmidrule{1-1}
\cmidrule{3-11}
MGR & \multirow{-4}{*}[3\dimexpr\aboverulesep+\belowrulesep+\cmidrulewidth]{\raggedright\arraybackslash 8000} & 0.0040 & 0.036 & 0.037 & 0.036 & 0.036 & 0.038 & 1.000 & 0.946 & 0.951\\
\bottomrule
\end{tabular}
\begin{tablenotes}
\item Abbreviations: ATE: average treatment effect; RR: relative risk; Est: Estimator; Med: Median; ASE: median asymptotic standard error; ESE: empirical standard error; MAD: mean absolute deviation; rRMSE: robust root mean squared error (RMSE) (square root of med bias squared + MAD squared); RMSE: RMSE (mean bias squared + ESE squared); Rel. Eff.: relative efficiency (rRMSE relative to that of MGR); Cover: coverage based on the ASE; ESE Cover: coverage based on the ESE; CGR: conditional generalized raking (GR); IPTW: inverse probability of treatment weighting; RA: regression adjustment; MGR: marginal GR; AIPCW: augmented inverse probability of coarsening weighting; AIPTW: augmented IPTW. Results are based on 2500 Monte-Carlo replications. 
\end{tablenotes}
\end{threeparttable}
\end{table}

% misspecified scenarios
\begin{table}[!h]
\centering
\caption{Performance of estimators of the ATE in Scenario 3a, a continuous outcome setting with approximately 50\% missing data and a misspecified treatment propensity model.\label{tab:sim_combined_performance_tab_s5_ATE}}
\centering
\fontsize{9}{11}\selectfont
\begin{threeparttable}
\begin{tabular}[t]{>{\raggedright\arraybackslash}p{3.5em}lrrrrrrrrr}
\toprule
Est. & n & Med. bias & ASE & ESE & MAD & rRMSE & RMSE & Rel. Eff. & Cover & ESE Cover\\
\midrule
AIPCW-AIPTW &  & -0.0008 & 0.227 & 0.235 & 0.234 & 0.234 & 0.235 & 0.867 & 0.947 & 0.951\\
\cmidrule{1-1}
\cmidrule{3-11}
CGR (IPTW) &  & -0.3703 & 0.410 & 0.401 & 0.396 & 0.542 & 0.549 & 0.374 & 0.857 & 0.852\\
\cmidrule{1-1}
\cmidrule{3-11}
CGR (RA) &  & -0.0129 & 0.189 & 0.198 & 0.198 & 0.198 & 0.199 & 1.024 & 0.937 & 0.947\\
\cmidrule{1-1}
\cmidrule{3-11}
MGR & \multirow{-4}{*}[3\dimexpr\aboverulesep+\belowrulesep+\cmidrulewidth]{\raggedright\arraybackslash 250} & -0.0077 & 0.228 & 0.206 & 0.203 & 0.203 & 0.206 & 1.000 & 0.970 & 0.953\\
\cmidrule{1-11}
AIPCW-AIPTW &  & -0.0007 & 0.139 & 0.155 & 0.152 & 0.152 & 0.155 & 0.932 & 0.924 & 0.952\\
\cmidrule{1-1}
\cmidrule{3-11}
CGR (IPTW) &  & -0.3657 & 0.286 & 0.278 & 0.278 & 0.460 & 0.460 & 0.309 & 0.750 & 0.735\\
\cmidrule{1-1}
\cmidrule{3-11}
CGR (RA) &  & -0.0030 & 0.135 & 0.141 & 0.139 & 0.139 & 0.141 & 1.024 & 0.940 & 0.946\\
\cmidrule{1-1}
\cmidrule{3-11}
MGR & \multirow{-4}{*}[3\dimexpr\aboverulesep+\belowrulesep+\cmidrulewidth]{\raggedright\arraybackslash 500} & 0.0009 & 0.139 & 0.144 & 0.142 & 0.142 & 0.144 & 1.000 & 0.942 & 0.950\\
\cmidrule{1-11}
AIPCW-AIPTW &  & -0.0025 & 0.091 & 0.103 & 0.103 & 0.103 & 0.103 & 0.953 & 0.915 & 0.949\\
\cmidrule{1-1}
\cmidrule{3-11}
CGR (IPTW) &  & -0.3673 & 0.200 & 0.195 & 0.193 & 0.415 & 0.418 & 0.237 & 0.546 & 0.526\\
\cmidrule{1-1}
\cmidrule{3-11}
CGR (RA) &  & -0.0013 & 0.095 & 0.097 & 0.096 & 0.096 & 0.097 & 1.025 & 0.946 & 0.949\\
\cmidrule{1-1}
\cmidrule{3-11}
MGR & \multirow{-4}{*}[3\dimexpr\aboverulesep+\belowrulesep+\cmidrulewidth]{\raggedright\arraybackslash 1000} & -0.0028 & 0.091 & 0.100 & 0.098 & 0.098 & 0.100 & 1.000 & 0.924 & 0.948\\
\cmidrule{1-11}
AIPCW-AIPTW &  & -0.0004 & 0.062 & 0.070 & 0.071 & 0.071 & 0.071 & 0.960 & 0.909 & 0.950\\
\cmidrule{1-1}
\cmidrule{3-11}
CGR (IPTW) &  & -0.3658 & 0.141 & 0.136 & 0.135 & 0.390 & 0.392 & 0.176 & 0.245 & 0.218\\
\cmidrule{1-1}
\cmidrule{3-11}
CGR (RA) &  & 0.0001 & 0.067 & 0.067 & 0.067 & 0.067 & 0.067 & 1.025 & 0.950 & 0.949\\
\cmidrule{1-1}
\cmidrule{3-11}
MGR & \multirow{-4}{*}[3\dimexpr\aboverulesep+\belowrulesep+\cmidrulewidth]{\raggedright\arraybackslash 2000} & -0.0009 & 0.062 & 0.069 & 0.068 & 0.068 & 0.069 & 1.000 & 0.921 & 0.950\\
\bottomrule
\end{tabular}
\begin{tablenotes}
\item Abbreviations: ATE: average treatment effect; RR: relative risk; RR: relative risk; Est: Estimator; Med: Median; ASE: median asymptotic standard error; ESE: empirical standard error; MAD: mean absolute deviation; rRMSE: robust root mean squared error (RMSE) (square root of med bias squared + MAD squared); RMSE: RMSE (mean bias squared + ESE squared); Rel. Eff.: relative efficiency (rRMSE relative to that of MGR); Cover: coverage based on the ASE; ESE Cover: coverage based on the ESE; CGR: conditional generalized raking (GR); IPTW: inverse probability of treatment weighting; RA: regression adjustment; MGR: marginal GR; AIPCW: augmented inverse probability of coarsening weighting; AIPTW: augmented IPTW. Results are based on 2500 Monte-Carlo replications. 
\end{tablenotes}
\end{threeparttable}
\end{table}

\begin{table}[!h]
\centering
\caption{Performance of estimators of the RR in Scenario 3b, a binary outcome setting with approximately 50\% missing data and a misspecified treatment propensity model.\label{tab:sim_combined_performance_tab_s6_RR}}
\centering
\fontsize{9}{11}\selectfont
\begin{threeparttable}
\begin{tabular}[t]{>{\raggedright\arraybackslash}p{3.5em}lrrrrrrrrr}
\toprule
Est. & n & Med. bias & ASE & ESE & MAD & rRMSE & RMSE & Rel. Eff. & Cover & ESE Cover\\
\midrule
AIPCW-AIPTW &  & -0.0031 & 0.187 & 0.192 & 0.184 & 0.184 & 0.194 & 0.993 & 0.939 & 0.946\\
\cmidrule{1-1}
\cmidrule{3-11}
CGR (RA) &  & -0.0017 & 0.154 & 0.172 & 0.171 & 0.171 & 0.175 & 1.069 & 0.914 & 0.950\\
\cmidrule{1-1}
\cmidrule{3-11}
MGR & \multirow{-3}{*}[2\dimexpr\aboverulesep+\belowrulesep+\cmidrulewidth]{\raggedright\arraybackslash 1000} & -0.0038 & 0.189 & 0.188 & 0.183 & 0.183 & 0.190 & 1.000 & 0.951 & 0.946\\
\cmidrule{1-11}
AIPCW-AIPTW &  & -0.0043 & 0.111 & 0.124 & 0.122 & 0.122 & 0.125 & 0.974 & 0.929 & 0.949\\
\cmidrule{1-1}
\cmidrule{3-11}
CGR (RA) &  & -0.0063 & 0.108 & 0.113 & 0.109 & 0.109 & 0.114 & 1.085 & 0.938 & 0.947\\
\cmidrule{1-1}
\cmidrule{3-11}
MGR & \multirow{-3}{*}[2\dimexpr\aboverulesep+\belowrulesep+\cmidrulewidth]{\raggedright\arraybackslash 2000} & -0.0046 & 0.112 & 0.122 & 0.119 & 0.119 & 0.123 & 1.000 & 0.932 & 0.946\\
\cmidrule{1-11}
AIPCW-AIPTW &  & 0.0025 & 0.072 & 0.085 & 0.086 & 0.086 & 0.086 & 0.993 & 0.903 & 0.946\\
\cmidrule{1-1}
\cmidrule{3-11}
CGR (RA) &  & 0.0013 & 0.076 & 0.079 & 0.080 & 0.080 & 0.079 & 1.062 & 0.941 & 0.950\\
\cmidrule{1-1}
\cmidrule{3-11}
MGR & \multirow{-3}{*}[2\dimexpr\aboverulesep+\belowrulesep+\cmidrulewidth]{\raggedright\arraybackslash 4000} & 0.0020 & 0.072 & 0.085 & 0.085 & 0.085 & 0.085 & 1.000 & 0.906 & 0.949\\
\cmidrule{1-11}
AIPCW-AIPTW &  & 0.0018 & 0.048 & 0.060 & 0.060 & 0.060 & 0.060 & 0.986 & 0.884 & 0.950\\
\cmidrule{1-1}
\cmidrule{3-11}
CGR (RA) &  & -0.0024 & 0.053 & 0.055 & 0.055 & 0.055 & 0.055 & 1.078 & 0.943 & 0.954\\
\cmidrule{1-1}
\cmidrule{3-11}
MGR & \multirow{-3}{*}[2\dimexpr\aboverulesep+\belowrulesep+\cmidrulewidth]{\raggedright\arraybackslash 8000} & 0.0007 & 0.048 & 0.060 & 0.059 & 0.059 & 0.060 & 1.000 & 0.880 & 0.952\\
\bottomrule
\end{tabular}
\begin{tablenotes}
\item Abbreviations: ATE: average treatment effect; RR: relative risk; RR: relative risk; RR: relative risk; Est: Estimator; Med: Median; ASE: median asymptotic standard error; ESE: empirical standard error; MAD: mean absolute deviation; rRMSE: robust root mean squared error (RMSE) (square root of med bias squared + MAD squared); RMSE: RMSE (mean bias squared + ESE squared); Rel. Eff.: relative efficiency (rRMSE relative to that of MGR); Cover: coverage based on the ASE; ESE Cover: coverage based on the ESE; CGR: conditional generalized raking (GR); IPTW: inverse probability of treatment weighting; RA: regression adjustment; MGR: marginal GR; AIPCW: augmented inverse probability of coarsening weighting; AIPTW: augmented IPTW. Results are based on 2500 Monte-Carlo replications. 
\end{tablenotes}
\end{threeparttable}
\end{table}

\begin{table}[!h]
\centering
\caption{Performance of CGR (IPTW) estimators of the ATE and RR in Scenario 3b, a binary outcome setting with approximately 50\% missing data and a misspecified treatment propensity model.\label{tab:sim_combined_performance_tab_s6_cgr_iptw}}
\centering
\fontsize{9}{11}\selectfont
\begin{threeparttable}
\begin{tabular}[t]{>{\raggedright\arraybackslash}p{3.5em}lrrrrrrrrr}
\toprule
Estimand & n & Med. bias & ASE & ESE & MAD & rRMSE & RMSE & Rel. Eff. & Cover & ESE Cover\\
\midrule
ATE & 1000 & 0.0370 & 0.020 & 0.020 & 0.020 & 0.042 & 0.042 & 0.388 & 0.538 & 0.531\\
\cmidrule{1-11}
ATE & 2000 & 0.0378 & 0.014 & 0.014 & 0.014 & 0.040 & 0.040 & 0.262 & 0.212 & 0.220\\
\cmidrule{1-11}
ATE & 4000 & 0.0378 & 0.010 & 0.010 & 0.010 & 0.039 & 0.039 & 0.191 & 0.027 & 0.030\\
\cmidrule{1-11}
ATE & 8000 & 0.0380 & 0.007 & 0.007 & 0.007 & 0.039 & 0.039 & 0.135 & 0.000 & 0.000\\
\cmidrule{1-11}
RR & 1000 & 0.7455 & 0.217 & 0.217 & 0.214 & 0.776 & 0.846 & 0.236 & 0.553 & 0.551\\
\cmidrule{1-11}
RR & 2000 & 0.7623 & 0.151 & 0.153 & 0.153 & 0.778 & 0.811 & 0.153 & 0.222 & 0.239\\
\cmidrule{1-11}
RR & 4000 & 0.7585 & 0.106 & 0.108 & 0.107 & 0.766 & 0.785 & 0.111 & 0.028 & 0.032\\
\cmidrule{1-11}
RR & 8000 & 0.7634 & 0.075 & 0.076 & 0.075 & 0.767 & 0.780 & 0.077 & 0.000 & 0.000\\
\bottomrule
\end{tabular}
\begin{tablenotes}
\item Abbreviations: Med: Median; ASE: median asymptotic standard error; ESE: empirical standard error; MAD: mean absolute deviation; rRMSE: robust root mean squared error (RMSE) (square root of med bias squared + MAD squared); RMSE: RMSE (mean bias squared + ESE squared); Rel. Eff.: relative efficiency (rRMSE relative to that of MGR); Cover: coverage based on the ASE; ESE Cover: coverage based on the ESE; CGR: conditional generalized raking (GR); IPTW: inverse probability of treatment weighting. Results are based on 2500 Monte-Carlo replications. 
\end{tablenotes}
\end{threeparttable}
\end{table}

\begin{table}[!h]
\centering
\caption{Performance of estimators of the ATE in Scenario 4a, a continuous outcome setting with approximately 50\% missing data and a misspecified outcome regression model.\label{tab:sim_combined_performance_tab_s7_ATE}}
\centering
\fontsize{9}{11}\selectfont
\begin{threeparttable}
\begin{tabular}[t]{>{\raggedright\arraybackslash}p{3.5em}lrrrrrrrrr}
\toprule
Est. & n & Med. bias & ASE & ESE & MAD & rRMSE & RMSE & Rel. Eff. & Cover & ESE Cover\\
\midrule
AIPCW-AIPTW &  & -0.0255 & 0.411 & 0.599 & 0.442 & 0.443 & 0.599 & 0.815 & 0.933 & 0.964\\
\cmidrule{1-1}
\cmidrule{3-11}
CGR (IPTW) &  & -0.0257 & 0.491 & 0.591 & 0.504 & 0.505 & 0.591 & 0.716 & 0.939 & 0.947\\
\cmidrule{1-1}
\cmidrule{3-11}
CGR (RA) &  & -0.3607 & 0.306 & 0.313 & 0.316 & 0.479 & 0.476 & 0.753 & 0.773 & 0.793\\
\cmidrule{1-1}
\cmidrule{3-11}
MGR & \multirow{-4}{*}[3\dimexpr\aboverulesep+\belowrulesep+\cmidrulewidth]{\raggedright\arraybackslash 250} & -0.0241 & 0.414 & 0.519 & 0.360 & 0.361 & 0.519 & 1.000 & 0.964 & 0.963\\
\cmidrule{1-11}
AIPCW-AIPTW &  & -0.0280 & 0.256 & 0.386 & 0.289 & 0.291 & 0.386 & 0.905 & 0.912 & 0.963\\
\cmidrule{1-1}
\cmidrule{3-11}
CGR (IPTW) &  & -0.0243 & 0.353 & 0.413 & 0.360 & 0.361 & 0.414 & 0.730 & 0.948 & 0.948\\
\cmidrule{1-1}
\cmidrule{3-11}
CGR (RA) &  & -0.3650 & 0.216 & 0.221 & 0.219 & 0.426 & 0.426 & 0.618 & 0.604 & 0.621\\
\cmidrule{1-1}
\cmidrule{3-11}
MGR & \multirow{-4}{*}[3\dimexpr\aboverulesep+\belowrulesep+\cmidrulewidth]{\raggedright\arraybackslash 500} & -0.0223 & 0.257 & 0.365 & 0.262 & 0.263 & 0.365 & 1.000 & 0.938 & 0.968\\
\cmidrule{1-11}
AIPCW-AIPTW &  & -0.0239 & 0.167 & 0.251 & 0.207 & 0.209 & 0.251 & 0.949 & 0.875 & 0.963\\
\cmidrule{1-1}
\cmidrule{3-11}
CGR (IPTW) &  & -0.0227 & 0.256 & 0.280 & 0.263 & 0.264 & 0.280 & 0.751 & 0.955 & 0.943\\
\cmidrule{1-1}
\cmidrule{3-11}
CGR (RA) &  & -0.3694 & 0.152 & 0.155 & 0.154 & 0.400 & 0.400 & 0.495 & 0.327 & 0.336\\
\cmidrule{1-1}
\cmidrule{3-11}
MGR & \multirow{-4}{*}[3\dimexpr\aboverulesep+\belowrulesep+\cmidrulewidth]{\raggedright\arraybackslash 1000} & -0.0181 & 0.167 & 0.240 & 0.197 & 0.198 & 0.240 & 1.000 & 0.904 & 0.966\\
\cmidrule{1-11}
AIPCW-AIPTW &  & -0.0159 & 0.113 & 0.157 & 0.141 & 0.142 & 0.157 & 0.974 & 0.866 & 0.956\\
\cmidrule{1-1}
\cmidrule{3-11}
CGR (IPTW) &  & -0.0089 & 0.186 & 0.197 & 0.193 & 0.193 & 0.197 & 0.716 & 0.954 & 0.950\\
\cmidrule{1-1}
\cmidrule{3-11}
CGR (RA) &  & -0.3665 & 0.107 & 0.111 & 0.110 & 0.383 & 0.381 & 0.361 & 0.079 & 0.087\\
\cmidrule{1-1}
\cmidrule{3-11}
MGR & \multirow{-4}{*}[3\dimexpr\aboverulesep+\belowrulesep+\cmidrulewidth]{\raggedright\arraybackslash 2000} & -0.0127 & 0.113 & 0.152 & 0.137 & 0.138 & 0.152 & 1.000 & 0.878 & 0.958\\
\bottomrule
\end{tabular}
\begin{tablenotes}
\item Abbreviations: ATE: average treatment effect; RR: relative risk; RR: relative risk; RR: relative risk; Est: Estimator; Med: Median; ASE: median asymptotic standard error; ESE: empirical standard error; MAD: mean absolute deviation; rRMSE: robust root mean squared error (RMSE) (square root of med bias squared + MAD squared); RMSE: RMSE (mean bias squared + ESE squared); Rel. Eff.: relative efficiency (rRMSE relative to that of MGR); Cover: coverage based on the ASE; ESE Cover: coverage based on the ESE; CGR: conditional generalized raking (GR); IPTW: inverse probability of treatment weighting; RA: regression adjustment; MGR: marginal GR; AIPCW: augmented inverse probability of coarsening weighting; AIPTW: augmented IPTW. Results are based on 2500 Monte-Carlo replications. 
\end{tablenotes}
\end{threeparttable}
\end{table}

\begin{table}[!h]
\centering
\caption{Performance of estimators of the RR in Scenario 4b, a binary outcome setting with approximately 50\% missing data and a misspecified outcome regression model.\label{tab:sim_combined_performance_tab_s8_RR}}
\centering
\fontsize{9}{11}\selectfont
\begin{threeparttable}
\begin{tabular}[t]{>{\raggedright\arraybackslash}p{3.5em}lrrrrrrrrr}
\toprule
Est. & n & Med. bias & ASE & ESE & MAD & rRMSE & RMSE & Rel. Eff. & Cover & ESE Cover\\
\midrule
AIPCW-AIPTW &  & 0.0530 & 0.266 & 0.378 & 0.341 & 0.345 & 0.396 & 0.957 & 0.880 & 0.959\\
\cmidrule{1-1}
\cmidrule{3-11}
CGR (RA) &  & 0.7825 & 0.220 & 0.216 & 0.211 & 0.811 & 0.869 & 0.408 & 0.526 & 0.514\\
\cmidrule{1-1}
\cmidrule{3-11}
MGR & \multirow{-3}{*}[2\dimexpr\aboverulesep+\belowrulesep+\cmidrulewidth]{\raggedright\arraybackslash 1000} & 0.0658 & 0.270 & 0.360 & 0.324 & 0.330 & 0.380 & 1.000 & 0.892 & 0.951\\
\cmidrule{1-11}
AIPCW-AIPTW &  & 0.0256 & 0.171 & 0.252 & 0.237 & 0.239 & 0.259 & 1.013 & 0.834 & 0.950\\
\cmidrule{1-1}
\cmidrule{3-11}
CGR (RA) &  & 0.7578 & 0.151 & 0.151 & 0.148 & 0.772 & 0.811 & 0.313 & 0.218 & 0.222\\
\cmidrule{1-1}
\cmidrule{3-11}
MGR & \multirow{-3}{*}[2\dimexpr\aboverulesep+\belowrulesep+\cmidrulewidth]{\raggedright\arraybackslash 2000} & 0.0345 & 0.172 & 0.249 & 0.239 & 0.242 & 0.257 & 1.000 & 0.847 & 0.951\\
\cmidrule{1-11}
AIPCW-AIPTW &  & 0.0092 & 0.114 & 0.180 & 0.175 & 0.175 & 0.182 & 1.005 & 0.806 & 0.952\\
\cmidrule{1-1}
\cmidrule{3-11}
CGR (RA) &  & 0.7664 & 0.106 & 0.105 & 0.104 & 0.773 & 0.783 & 0.228 & 0.028 & 0.031\\
\cmidrule{1-1}
\cmidrule{3-11}
MGR & \multirow{-3}{*}[2\dimexpr\aboverulesep+\belowrulesep+\cmidrulewidth]{\raggedright\arraybackslash 4000} & 0.0185 & 0.114 & 0.180 & 0.175 & 0.176 & 0.183 & 1.000 & 0.804 & 0.949\\
\cmidrule{1-11}
AIPCW-AIPTW &  & 0.0031 & 0.078 & 0.133 & 0.126 & 0.126 & 0.133 & 0.986 & 0.778 & 0.957\\
\cmidrule{1-1}
\cmidrule{3-11}
CGR (RA) &  & 0.7559 & 0.074 & 0.072 & 0.073 & 0.759 & 0.769 & 0.163 & 0.000 & 0.000\\
\cmidrule{1-1}
\cmidrule{3-11}
MGR & \multirow{-3}{*}[2\dimexpr\aboverulesep+\belowrulesep+\cmidrulewidth]{\raggedright\arraybackslash 8000} & 0.0095 & 0.078 & 0.132 & 0.123 & 0.124 & 0.133 & 1.000 & 0.778 & 0.956\\
\bottomrule
\end{tabular}
\begin{tablenotes}
\item Abbreviations: ATE: average treatment effect; RR: relative risk; RR: relative risk; RR: relative risk; RR: relative risk; Est: Estimator; Med: Median; ASE: median asymptotic standard error; ESE: empirical standard error; MAD: mean absolute deviation; rRMSE: robust root mean squared error (RMSE) (square root of med bias squared + MAD squared); RMSE: RMSE (mean bias squared + ESE squared); Rel. Eff.: relative efficiency (rRMSE relative to that of MGR); Cover: coverage based on the ASE; ESE Cover: coverage based on the ESE; CGR: conditional generalized raking (GR); IPTW: inverse probability of treatment weighting; RA: regression adjustment; MGR: marginal GR; AIPCW: augmented inverse probability of coarsening weighting; AIPTW: augmented IPTW. Results are based on 2500 Monte-Carlo replications. 
\end{tablenotes}
\end{threeparttable}
\end{table}

\begin{table}[!h]
\centering
\caption{Performance of CGR (IPTW) estimators of the ATE and RR in Scenario 4b, a binary outcome setting with approximately 50\% missing data and a misspecified outcome regression model.\label{tab:sim_combined_performance_tab_s8_cgr_iptw}}
\centering
\fontsize{9}{11}\selectfont
\begin{threeparttable}
\begin{tabular}[t]{>{\raggedright\arraybackslash}p{3.5em}lrrrrrrrrr}
\toprule
Estimand & n & Med. bias & ASE & ESE & MAD & rRMSE & RMSE & Rel. Eff. & Cover & ESE Cover\\
\midrule
ATE & 1000 & 0.0004 & 0.023 & 0.027 & 0.025 & 0.025 & 0.027 & 1.069 & 0.950 & 0.950\\
\cmidrule{1-11}
ATE & 2000 & 0.0002 & 0.017 & 0.019 & 0.019 & 0.019 & 0.019 & 1.090 & 0.944 & 0.952\\
\cmidrule{1-11}
ATE & 4000 & $<$ 0.0001 & 0.012 & 0.013 & 0.013 & 0.013 & 0.013 & 1.174 & 0.948 & 0.949\\
\cmidrule{1-11}
ATE & 8000 & 0.0002 & 0.009 & 0.010 & 0.009 & 0.009 & 0.010 & 1.144 & 0.945 & 0.952\\
\cmidrule{1-11}
RR & 1000 & 0.0217 & 0.265 & 0.306 & 0.308 & 0.308 & 0.318 & 1.072 & 0.930 & 0.947\\
\cmidrule{1-11}
RR & 2000 & 0.0050 & 0.192 & 0.215 & 0.213 & 0.213 & 0.219 & 1.134 & 0.927 & 0.953\\
\cmidrule{1-11}
RR & 4000 & 0.0008 & 0.137 & 0.149 & 0.149 & 0.149 & 0.151 & 1.185 & 0.937 & 0.948\\
\cmidrule{1-11}
RR & 8000 & 0.0011 & 0.099 & 0.110 & 0.109 & 0.109 & 0.111 & 1.135 & 0.938 & 0.956\\
\bottomrule
\end{tabular}
\begin{tablenotes}
\item Abbreviations: Med: Median; ASE: median asymptotic standard error; ESE: empirical standard error; MAD: mean absolute deviation; rRMSE: robust root mean squared error (RMSE) (square root of med bias squared + MAD squared); RMSE: RMSE (mean bias squared + ESE squared); Rel. Eff.: relative efficiency (rRMSE relative to that of MGR); Cover: coverage based on the ASE; ESE Cover: coverage based on the ESE; CGR: conditional generalized raking (GR); IPTW: inverse probability of treatment weighting. Results are based on 2500 Monte-Carlo replications. 
\end{tablenotes}
\end{threeparttable}
\end{table}

\begin{table}[!h]
\centering
\caption{Performance of estimators of the ATE in Scenario 5a, a continuous outcome setting with approximately 50\% missing data and a misspecified missing-data model.\label{tab:sim_combined_performance_tab_s9_ATE}}
\centering
\fontsize{9}{11}\selectfont
\begin{threeparttable}
\begin{tabular}[t]{>{\raggedright\arraybackslash}p{3.5em}lrrrrrrrrr}
\toprule
Est. & n & Med. bias & ASE & ESE & MAD & rRMSE & RMSE & Rel. Eff. & Cover & ESE Cover\\
\midrule
AIPCW-AIPTW &  & -0.0073 & 0.210 & 0.204 & 0.208 & 0.208 & 0.204 & 0.899 & 0.959 & 0.949\\
\cmidrule{1-1}
\cmidrule{3-11}
CGR (IPTW) &  & -0.0228 & 0.342 & 0.349 & 0.335 & 0.336 & 0.349 & 0.557 & 0.964 & 0.948\\
\cmidrule{1-1}
\cmidrule{3-11}
CGR (RA) &  & -0.0020 & 0.171 & 0.182 & 0.182 & 0.182 & 0.182 & 1.027 & 0.934 & 0.952\\
\cmidrule{1-1}
\cmidrule{3-11}
MGR & \multirow{-4}{*}[3\dimexpr\aboverulesep+\belowrulesep+\cmidrulewidth]{\raggedright\arraybackslash 250} & -0.0042 & 0.210 & 0.187 & 0.187 & 0.187 & 0.188 & 1.000 & 0.977 & 0.957\\
\cmidrule{1-11}
AIPCW-AIPTW &  & -0.0061 & 0.137 & 0.128 & 0.130 & 0.130 & 0.128 & 0.945 & 0.963 & 0.952\\
\cmidrule{1-1}
\cmidrule{3-11}
CGR (IPTW) &  & 0.0008 & 0.237 & 0.233 & 0.234 & 0.234 & 0.233 & 0.526 & 0.966 & 0.950\\
\cmidrule{1-1}
\cmidrule{3-11}
CGR (RA) &  & -0.0049 & 0.121 & 0.121 & 0.125 & 0.125 & 0.121 & 0.986 & 0.949 & 0.950\\
\cmidrule{1-1}
\cmidrule{3-11}
MGR & \multirow{-4}{*}[3\dimexpr\aboverulesep+\belowrulesep+\cmidrulewidth]{\raggedright\arraybackslash 500} & -0.0053 & 0.137 & 0.123 & 0.123 & 0.123 & 0.123 & 1.000 & 0.970 & 0.952\\
\cmidrule{1-11}
AIPCW-AIPTW &  & -0.0003 & 0.092 & 0.089 & 0.088 & 0.088 & 0.089 & 0.960 & 0.957 & 0.948\\
\cmidrule{1-1}
\cmidrule{3-11}
CGR (IPTW) &  & 0.0053 & 0.166 & 0.163 & 0.160 & 0.161 & 0.163 & 0.526 & 0.955 & 0.948\\
\cmidrule{1-1}
\cmidrule{3-11}
CGR (RA) &  & -0.0010 & 0.086 & 0.085 & 0.085 & 0.085 & 0.085 & 0.990 & 0.949 & 0.948\\
\cmidrule{1-1}
\cmidrule{3-11}
MGR & \multirow{-4}{*}[3\dimexpr\aboverulesep+\belowrulesep+\cmidrulewidth]{\raggedright\arraybackslash 1000} & -0.0028 & 0.091 & 0.086 & 0.084 & 0.084 & 0.086 & 1.000 & 0.963 & 0.949\\
\cmidrule{1-11}
AIPCW-AIPTW &  & 0.0008 & 0.063 & 0.061 & 0.063 & 0.063 & 0.061 & 0.950 & 0.956 & 0.948\\
\cmidrule{1-1}
\cmidrule{3-11}
CGR (IPTW) &  & 0.0041 & 0.117 & 0.112 & 0.115 & 0.115 & 0.112 & 0.520 & 0.964 & 0.955\\
\cmidrule{1-1}
\cmidrule{3-11}
CGR (RA) &  & 0.0002 & 0.061 & 0.061 & 0.061 & 0.061 & 0.061 & 0.975 & 0.950 & 0.947\\
\cmidrule{1-1}
\cmidrule{3-11}
MGR & \multirow{-4}{*}[3\dimexpr\aboverulesep+\belowrulesep+\cmidrulewidth]{\raggedright\arraybackslash 2000} & 0.0002 & 0.063 & 0.060 & 0.060 & 0.060 & 0.060 & 1.000 & 0.961 & 0.950\\
\bottomrule
\end{tabular}
\begin{tablenotes}
\item Abbreviations: ATE: average treatment effect; RR: relative risk; RR: relative risk; RR: relative risk; RR: relative risk; Est: Estimator; Med: Median; ASE: median asymptotic standard error; ESE: empirical standard error; MAD: mean absolute deviation; rRMSE: robust root mean squared error (RMSE) (square root of med bias squared + MAD squared); RMSE: RMSE (mean bias squared + ESE squared); Rel. Eff.: relative efficiency (rRMSE relative to that of MGR); Cover: coverage based on the ASE; ESE Cover: coverage based on the ESE; CGR: conditional generalized raking (GR); IPTW: inverse probability of treatment weighting; RA: regression adjustment; MGR: marginal GR; AIPCW: augmented inverse probability of coarsening weighting; AIPTW: augmented IPTW. Results are based on 2500 Monte-Carlo replications. 
\end{tablenotes}
\end{threeparttable}
\end{table}

\begin{table}[!h]
\centering
\caption{Performance of estimators of the ATE in Scenario 5b, a binary outcome setting with approximately 50\% missing data and a misspecified missing-data model.\label{tab:sim_combined_performance_tab_s10_ATE}}
\centering
\fontsize{9}{11}\selectfont
\begin{threeparttable}
\begin{tabular}[t]{>{\raggedright\arraybackslash}p{3.5em}lrrrrrrrrr}
\toprule
Est. & n & Med. bias & ASE & ESE & MAD & rRMSE & RMSE & Rel. Eff. & Cover & ESE Cover\\
\midrule
AIPCW-AIPTW &  & -5 $\times 10^{-04}$ & 0.021 & 0.020 & 0.021 & 0.021 & 0.020 & 0.976 & 0.964 & 0.955\\
\cmidrule{1-1}
\cmidrule{3-11}
CGR (IPTW) &  & -2 $\times 10^{-04}$ & 0.021 & 0.021 & 0.021 & 0.021 & 0.021 & 0.937 & 0.951 & 0.951\\
\cmidrule{1-1}
\cmidrule{3-11}
CGR (RA) &  & -5 $\times 10^{-04}$ & 0.019 & 0.019 & 0.020 & 0.020 & 0.019 & 1.001 & 0.948 & 0.958\\
\cmidrule{1-1}
\cmidrule{3-11}
MGR & \multirow{-4}{*}[3\dimexpr\aboverulesep+\belowrulesep+\cmidrulewidth]{\raggedright\arraybackslash 1000} & -1 $\times 10^{-03}$ & 0.021 & 0.020 & 0.020 & 0.020 & 0.020 & 1.000 & 0.969 & 0.952\\
\cmidrule{1-11}
AIPCW-AIPTW &  & 7 $\times 10^{-04}$ & 0.014 & 0.014 & 0.014 & 0.015 & 0.014 & 0.953 & 0.941 & 0.945\\
\cmidrule{1-1}
\cmidrule{3-11}
CGR (IPTW) &  & -3 $\times 10^{-04}$ & 0.015 & 0.015 & 0.014 & 0.014 & 0.015 & 0.967 & 0.946 & 0.946\\
\cmidrule{1-1}
\cmidrule{3-11}
CGR (RA) &  & 4 $\times 10^{-04}$ & 0.013 & 0.014 & 0.014 & 0.014 & 0.014 & 0.997 & 0.948 & 0.954\\
\cmidrule{1-1}
\cmidrule{3-11}
MGR & \multirow{-4}{*}[3\dimexpr\aboverulesep+\belowrulesep+\cmidrulewidth]{\raggedright\arraybackslash 2000} & 3 $\times 10^{-04}$ & 0.014 & 0.014 & 0.014 & 0.014 & 0.014 & 1.000 & 0.955 & 0.954\\
\cmidrule{1-11}
AIPCW-AIPTW &  & 9 $\times 10^{-04}$ & 0.010 & 0.010 & 0.010 & 0.010 & 0.010 & 0.949 & 0.938 & 0.944\\
\cmidrule{1-1}
\cmidrule{3-11}
CGR (IPTW) &  & $<$ 0.0001 & 0.010 & 0.010 & 0.010 & 0.010 & 0.010 & 0.920 & 0.948 & 0.944\\
\cmidrule{1-1}
\cmidrule{3-11}
CGR (RA) &  & 2 $\times 10^{-04}$ & 0.009 & 0.009 & 0.009 & 0.009 & 0.009 & 1.025 & 0.950 & 0.950\\
\cmidrule{1-1}
\cmidrule{3-11}
MGR & \multirow{-4}{*}[3\dimexpr\aboverulesep+\belowrulesep+\cmidrulewidth]{\raggedright\arraybackslash 4000} & $<$ 0.0001 & 0.010 & 0.009 & 0.009 & 0.009 & 0.009 & 1.000 & 0.951 & 0.948\\
\cmidrule{1-11}
AIPCW-AIPTW &  & 5 $\times 10^{-04}$ & 0.007 & 0.007 & 0.007 & 0.007 & 0.007 & 0.978 & 0.937 & 0.952\\
\cmidrule{1-1}
\cmidrule{3-11}
CGR (IPTW) &  & -3 $\times 10^{-04}$ & 0.007 & 0.007 & 0.007 & 0.007 & 0.007 & 0.929 & 0.948 & 0.949\\
\cmidrule{1-1}
\cmidrule{3-11}
CGR (RA) &  & -1 $\times 10^{-04}$ & 0.007 & 0.007 & 0.007 & 0.007 & 0.007 & 1.012 & 0.940 & 0.950\\
\cmidrule{1-1}
\cmidrule{3-11}
MGR & \multirow{-4}{*}[3\dimexpr\aboverulesep+\belowrulesep+\cmidrulewidth]{\raggedright\arraybackslash 8000} & -1 $\times 10^{-04}$ & 0.007 & 0.007 & 0.007 & 0.007 & 0.007 & 1.000 & 0.943 & 0.952\\
\bottomrule
\end{tabular}
\begin{tablenotes}
\item Abbreviations: ATE: average treatment effect; RR: relative risk; RR: relative risk; RR: relative risk; RR: relative risk; RR: relative risk; Est: Estimator; Med: Median; ASE: median asymptotic standard error; ESE: empirical standard error; MAD: mean absolute deviation; rRMSE: robust root mean squared error (RMSE) (square root of med bias squared + MAD squared); RMSE: RMSE (mean bias squared + ESE squared); Rel. Eff.: relative efficiency (rRMSE relative to that of MGR); Cover: coverage based on the ASE; ESE Cover: coverage based on the ESE; CGR: conditional generalized raking (GR); IPTW: inverse probability of treatment weighting; RA: regression adjustment; MGR: marginal GR; AIPCW: augmented inverse probability of coarsening weighting; AIPTW: augmented IPTW. Results are based on 2500 Monte-Carlo replications. 
\end{tablenotes}
\end{threeparttable}
\end{table}

\begin{table}[!h]
\centering
\caption{Performance of estimators of the RR in Scenario 5b, a binary outcome setting with approximately 50\% missing data and a misspecified missing-data model.\label{tab:sim_combined_performance_tab_s10_RR}}
\centering
\fontsize{9}{11}\selectfont
\begin{threeparttable}
\begin{tabular}[t]{>{\raggedright\arraybackslash}p{3.5em}lrrrrrrrrr}
\toprule
Est. & n & Med. bias & ASE & ESE & MAD & rRMSE & RMSE & Rel. Eff. & Cover & ESE Cover\\
\midrule
AIPCW-AIPTW &  & -0.0125 & 0.228 & 0.209 & 0.208 & 0.208 & 0.217 & 0.972 & 0.971 & 0.952\\
\cmidrule{1-1}
\cmidrule{3-11}
CGR (IPTW) &  & 0.0140 & 0.213 & 0.217 & 0.217 & 0.217 & 0.232 & 0.933 & 0.951 & 0.952\\
\cmidrule{1-1}
\cmidrule{3-11}
CGR (RA) &  & 0.0180 & 0.196 & 0.200 & 0.205 & 0.205 & 0.214 & 0.986 & 0.949 & 0.950\\
\cmidrule{1-1}
\cmidrule{3-11}
MGR & \multirow{-4}{*}[3\dimexpr\aboverulesep+\belowrulesep+\cmidrulewidth]{\raggedright\arraybackslash 1000} & -0.0010 & 0.230 & 0.202 & 0.202 & 0.202 & 0.214 & 1.000 & 0.976 & 0.950\\
\cmidrule{1-11}
AIPCW-AIPTW &  & 0.0053 & 0.149 & 0.145 & 0.143 & 0.143 & 0.148 & 0.985 & 0.958 & 0.947\\
\cmidrule{1-1}
\cmidrule{3-11}
CGR (IPTW) &  & 0.0036 & 0.149 & 0.149 & 0.144 & 0.144 & 0.154 & 0.977 & 0.944 & 0.943\\
\cmidrule{1-1}
\cmidrule{3-11}
CGR (RA) &  & 0.0167 & 0.137 & 0.139 & 0.136 & 0.137 & 0.145 & 1.026 & 0.946 & 0.948\\
\cmidrule{1-1}
\cmidrule{3-11}
MGR & \multirow{-4}{*}[3\dimexpr\aboverulesep+\belowrulesep+\cmidrulewidth]{\raggedright\arraybackslash 2000} & 0.0152 & 0.149 & 0.142 & 0.140 & 0.141 & 0.147 & 1.000 & 0.965 & 0.948\\
\cmidrule{1-11}
AIPCW-AIPTW &  & -0.0071 & 0.101 & 0.096 & 0.094 & 0.095 & 0.097 & 0.982 & 0.957 & 0.946\\
\cmidrule{1-1}
\cmidrule{3-11}
CGR (IPTW) &  & 0.0066 & 0.105 & 0.101 & 0.095 & 0.096 & 0.104 & 0.971 & 0.958 & 0.946\\
\cmidrule{1-1}
\cmidrule{3-11}
CGR (RA) &  & 0.0125 & 0.096 & 0.093 & 0.091 & 0.092 & 0.096 & 1.008 & 0.956 & 0.946\\
\cmidrule{1-1}
\cmidrule{3-11}
MGR & \multirow{-4}{*}[3\dimexpr\aboverulesep+\belowrulesep+\cmidrulewidth]{\raggedright\arraybackslash 4000} & 0.0069 & 0.101 & 0.094 & 0.093 & 0.093 & 0.096 & 1.000 & 0.965 & 0.948\\
\cmidrule{1-11}
AIPCW-AIPTW &  & -0.0175 & 0.070 & 0.071 & 0.072 & 0.074 & 0.071 & 0.943 & 0.949 & 0.952\\
\cmidrule{1-1}
\cmidrule{3-11}
CGR (IPTW) &  & -0.0025 & 0.074 & 0.074 & 0.073 & 0.074 & 0.074 & 0.949 & 0.949 & 0.951\\
\cmidrule{1-1}
\cmidrule{3-11}
CGR (RA) &  & -0.0045 & 0.067 & 0.069 & 0.069 & 0.070 & 0.069 & 1.002 & 0.948 & 0.952\\
\cmidrule{1-1}
\cmidrule{3-11}
MGR & \multirow{-4}{*}[3\dimexpr\aboverulesep+\belowrulesep+\cmidrulewidth]{\raggedright\arraybackslash 8000} & -0.0049 & 0.070 & 0.069 & 0.070 & 0.070 & 0.069 & 1.000 & 0.958 & 0.952\\
\bottomrule
\end{tabular}
\begin{tablenotes}
\item Abbreviations: ATE: average treatment effect; RR: relative risk; RR: relative risk; RR: relative risk; RR: relative risk; RR: relative risk; Est: Estimator; Med: Median; ASE: median asymptotic standard error; ESE: empirical standard error; MAD: mean absolute deviation; rRMSE: robust root mean squared error (RMSE) (square root of med bias squared + MAD squared); RMSE: RMSE (mean bias squared + ESE squared); Rel. Eff.: relative efficiency (rRMSE relative to that of MGR); Cover: coverage based on the ASE; ESE Cover: coverage based on the ESE; CGR: conditional generalized raking (GR); IPTW: inverse probability of treatment weighting; RA: regression adjustment; MGR: marginal GR; AIPCW: augmented inverse probability of coarsening weighting; AIPTW: augmented IPTW. Results are based on 2500 Monte-Carlo replications. 
\end{tablenotes}
\end{threeparttable}
\end{table}

\section{Additional results from the VCCC analysis}

In Table~\ref{tab:vccc_num_prop_errors}, we present the number and proportion of records that had an error for each error-prone variable and overall. In Table~\ref{tab:vccc_unvalidated_class_perf}, we display the number of true and false positives and true and false negatives when using the unvalidated variable (either the outcome, AIDS-defining event or death within 5 years of treatment initiation, or the exposure, starting a PI-containing regimen) as a predictor of the validated variable. We also present sensitivity, specificity, positive predictive value (PPV), and negative predictive value (NPV).

\begin{table}[!h]
\centering
\caption{Number and proportion of errors in variables in the VCCC data. \label{tab:vccc_num_prop_errors}}
\centering
\fontsize{9}{11}\selectfont
\begin{threeparttable}
\begin{tabular}[t]{>{\raggedright\arraybackslash}p{30em}r}
\toprule
Variable & N (percent)\\
\midrule
AIDS-defining event or death within 5 years & 57 (9\%)\\
CD4 count & 39 (6\%)\\
Year of ART initiation & 20 (3\%)\\
Initiated PI-containing ART regimen & 9 (1\%)\\
Error in any variable & 97 (15\%)\\
\bottomrule
\end{tabular}
\begin{tablenotes}
\item Abbreviations: ART: antiretroviral treatment; PI: protease inhibitor.
\end{tablenotes}
\end{threeparttable}
\end{table}

\begin{table}[!h]
\centering
\caption{Classification performance of the unvalidated variables in the VCCC data. \label{tab:vccc_unvalidated_class_perf}}
\centering
\fontsize{9}{11}\selectfont
\begin{threeparttable}
\begin{tabular}[t]{>{\raggedright\arraybackslash}p{15em}rrrrrrrr}
\toprule
Variable & TP & TN & FP & FN & Sensitivity & Specificity & PPV & NPV\\
\midrule
AIDS-defining event or death within 5 years & 58 & 513 & 54 & 3 & 0.951 & 0.905 & 0.518 & 0.994\\
PI-containing ART regiment & 262 & 357 & 5 & 4 & 0.985 & 0.986 & 0.981 & 0.989\\
\bottomrule
\end{tabular}
\begin{tablenotes}
\item Abbreviations: TP: true positive; TN: true negative; FP: false positive; FN: false negative; PPV: positive predictive value; NPV: negative predictive value; ART: antiretroviral treatment; PI: protease inhibitor.
\end{tablenotes}
\end{threeparttable}
\end{table}

In Table~\ref{tab:vccc_descriptive}, we present descriptive statistics for the Vanderbilt Comprehensive Care Clinic (VCCC) data analysis, by source dataset (unvalidated or validated data).

\begin{table}

\caption{\label{tab:vccc_descriptive}Descriptive statistics for the VCCC data, by source (unvalidated vs validated)}
\centering
\begin{tabular}[t]{lcc}
\toprule
\multicolumn{1}{c}{ } & \multicolumn{2}{c}{\textbf{Source}} \\
\cmidrule(l{3pt}r{3pt}){2-3}
\textbf{Characteristic} & \makecell[c]{\textbf{Unvalidated}\\N = 628} & \makecell[c]{\textbf{Validated}\\N = 628}\\
\midrule
Age at treatment initiation & 37.3 (10.1) & 37.3 (10.1)\\
Self-reported Black race & 273 (43\%) & 273 (43\%)\\
Self-reported injection drug use & 50 (8.0\%) & 50 (8.0\%)\\
Female & 164 (26\%) & 164 (26\%)\\
Year of treatment initiation &  & \\
\hspace{1em}1998 & 29 (4.6\%) & 25 (4.0\%)\\
\hspace{1em}1999 & 27 (4.3\%) & 29 (4.6\%)\\
\hspace{1em}2000 & 40 (6.4\%) & 40 (6.4\%)\\
\hspace{1em}2001 & 31 (4.9\%) & 32 (5.1\%)\\
\hspace{1em}2002 & 37 (5.9\%) & 39 (6.2\%)\\
\hspace{1em}2003 & 31 (4.9\%) & 29 (4.6\%)\\
\hspace{1em}2004 & 45 (7.2\%) & 45 (7.2\%)\\
\hspace{1em}2005 & 31 (4.9\%) & 30 (4.8\%)\\
\hspace{1em}2006 & 52 (8.3\%) & 52 (8.3\%)\\
\hspace{1em}2007 & 39 (6.2\%) & 40 (6.4\%)\\
\hspace{1em}2008 & 63 (10\%) & 62 (9.9\%)\\
\hspace{1em}2009 & 92 (15\%) & 89 (14\%)\\
\hspace{1em}2010 & 49 (7.8\%) & 53 (8.4\%)\\
\hspace{1em}2011 & 32 (5.1\%) & 34 (5.4\%)\\
\hspace{1em}2012 & 19 (3.0\%) & 18 (2.9\%)\\
\hspace{1em}2013 & 11 (1.8\%) & 11 (1.8\%)\\
CD4 count at treatment initiation (square root) & 16.1 (5.9) & 16.0 (5.9)\\
Protease inhibitor-containing regimen & 267 (43\%) & 266 (42\%)\\
ADE or death within 5 years of treatment initiation & 112 (18\%) & 61 (9.7\%)\\
\bottomrule
\multicolumn{3}{l}{\rule{0pt}{1em}\textsuperscript{1} For categorical variables, we report n (\%); for continuous variables, we report mean (SD).}\\
\end{tabular}
\end{table}

\clearpage
\fi

{\small
\bibliographystyle{chicago}
\bibliography{papers}
}

\end{document}